\documentclass[a4paper,10pt]{article}

\usepackage[utf8]{inputenc}

\usepackage{xurl}
\usepackage[colorlinks=true,linkcolor=blue,citecolor=blue]{hyperref}
\usepackage{geometry}
\usepackage{amsmath}
\usepackage{amssymb}
\usepackage{amsthm}
\usepackage{indentfirst}
\usepackage{wrapfig}
\usepackage{graphicx}
\usepackage{mathtools}
\usepackage{relsize}

\newtheorem{Theorem}{Theorem}[section]  
\newtheorem{Lemma}[Theorem]{Lemma}
\newtheorem{Corollary}[Theorem]{Corollary}
\newtheorem{Proposition}[Theorem]{Proposition}

\newtheorem{Definition}[Theorem]{Definition}
\newtheorem{Remark}[Theorem]{Remark}
\newtheorem{Examples}[Theorem]{Examples}

\newtheorem*{notation}{Notation}

\newcommand{\R}{{\mathbb R}}
\newcommand{\N}{{\mathbb N}}

\newcommand{\PP}{{\mathcal P}}
\newcommand{\E}{\mathcal E}

\newcommand{\inn}{\operatorname{in}}
\newcommand{\out}{ \operatorname{out} }
\def\cc#1{\mathcal{#1}}
\def\xto#1{\xrightarrow{#1}}
\def\xtoo#1{\xlongrightarrow{#1}}
\def\xlongto#1{\stackrel{#1}{\longrightarrow}}

\newcommand{\sbt}{\,\begin{picture}(-1,1)(-1,-3)\circle*{3}\end{picture}\ }

\title{Walraswap: equilibrium-based decentralized exchange integrating external AMM liquidity
}
\author{Sergio Andrés Yuhjtman}
\date{}

\hypersetup{
  pdftitle={Walraswap: equilibrium-based decentralized exchange integrating external AMM liquidity},
  pdfauthor={Sergio Andrés Yuhjtman}
}

\begin{document}
\maketitle

\begin{abstract}
This article advocates the development of a decentralized exchange protocol for Ethereum based on Walrasian equilibrium.
Building on previous work \cite{Speedex}, \cite{Augmenting}, it develops a mathematical framework that incorporates liquidity from independent automated market makers (AMMs). Trade orders are represented by continuous monotonic pairwise supply functions
and external AMMs are represented by concave swap functions. Under a connectivity condition, we establish the existence of a common clearing-price vector and external swaps that jointly fulfill the prescribed orders while leaving a nonnegative surplus of every token. External AMMs execute ordinary swaps according to their own pricing rules, including swap fees. For the selected clearing-price vector, the prescribed external swaps maximize the value of the batch surplus. We also establish structural properties of equilibrium, including invariance of individual net trades across equilibria, and formulate equilibrium computation as a convex optimization problem. We then outline Walraswap, an exchange protocol for Ethereum with off-chain batch construction and on-chain settlement. Its proposed governance and solver accountability mechanisms aim to encourage public order disclosure and discourage censorship.
The protocol design explores how the mathematical construction could support decentralized trading with reduced opportunities for {\it maximal extractable value}.
\end{abstract}

\section{Introduction}

Starting in 2019, decentralized exchange (DEX) activity on public blockchains grew significantly, most notably on Ethereum, driven by the adoption of automated market makers (AMMs)
\cite{Univ2}, \cite{Bal}, \cite{Curve}, \cite{Univ3}.
The surge of decentralized finance gave rise to a complex and interesting type of activity known as {\it maximal extractable value} (MEV) \cite{F2}.
Block producers have the power to organize the transactions in their blocks at will. Since the effect of swaps operating on a given AMM depends on the ordering, this leads to a significant economic advantage for the block producer. The Ethereum community soon recognized that MEV had created pressure toward centralization among block producers, who were incentivized to vertically integrate with trading firms to boost their returns. In part to mitigate that pressure, MEV-Boost was introduced: a system that allows block proposers to outsource the building of blocks~\cite{MEVBoostGitHub}. This out-of-protocol implementation of proposer-builder separation (PBS, \cite{EthereumPBS}) has since been widely adopted on Ethereum and has contributed to democratizing access to MEV revenue among block proposers \cite{WhyMEVBoost}.

Another response to MEV has been to design DEX protocols that reduce its incidence by leveraging the expressive power of smart contracts.

A notable contributor to this quest has been CoW Protocol.
Originally, CoW advocated and implemented uniform clearing price batches
following the rule ``one price per token per block'' \cite{LeupoldFairerDEXs}, \cite{Walther}, as a means of achieving fairer decentralized exchange with greater trading volume. In fact, such a structure enjoys these two fundamental properties:

\medskip
(a) The net effect of a batch can be independent of the ordering of its constituent operations.

(b) Uniform clearing prices eliminate price spreads between trades in the batch. Buying and reselling a token within the batch yields no gain from differences in exchange
rates.

\medskip
This idea was also adopted by SPEEDEX \cite{Speedex} and further studied in \cite{Augmenting} by Ramseyer, Goyal, Goel, and Mazières. SPEEDEX proposes a multi-asset batch exchange as a repeated Arrow-Debreu (or Walrasian) pure exchange market, developing practical equilibrium-approximation algorithms and an architecture for parallel transaction processing.
``Augmenting Batch Exchanges with Constant Function Market Makers'' \cite{Augmenting} develops an axiomatic treatment of the integration of constant-function market makers into such exchanges and identifies trade-offs among natural market-clearing rules. Of particular interest is the formulation of a convex program for efficiently computing market equilibria in a sufficiently general case \cite[\S 5.2]{Augmenting}.

 The present article continues along the road ``towards an ideal DEX''~\cite[\S 1.1]{Speedex} by addressing two main aspects beyond the scope of \cite{Speedex} and \cite{Augmenting}. The first is interoperability between such a DEX and independent external AMMs. Can a uniform-price batch exchange systematically use liquidity from independently deployed AMMs,  without requiring those AMMs to trade at the batch's clearing prices, and without creating token deficits in the batch? Our main rigorous result, Theorem \ref{thm:main}, gives a positive answer.
A uniform clearing-price vector can be selected together with external-AMM swaps so that the resulting batch is feasible and yields a nonnegative surplus of every token. Moreover, the prescribed choice of external-AMM swaps maximizes the surplus when measured according to the batch price vector (Remark \ref{remark:max-surplus}). Theorem \ref{thm:main} is closely related to \cite[Theorem C.6]{Augmenting}; see Remark \ref{remark:closely-related}.

Theorem \ref{thm:main} provides the mathematical basis for an exchange mechanism. Our second contribution explores how such a mechanism could be implemented on Ethereum.
In Section \ref{section:mechanism}, we outline a concrete decentralized exchange protocol, Walraswap, designed for this purpose. The proposed solution is based on a governance mechanism for selecting solvers that encourages full disclosure of trade orders. This section does not contain any rigorous results, and the proposed protocol is not expected to be implemented without modifications.
Walraswap is the result of an effort to develop an exchange protocol on Ethereum that executes, block by block, an approximate Walrasian equilibrium among participating agents.
Theorem \ref{thm:main} provides a mathematical tool that could help transform the current fragmented trading environment into a highly liquid approximate Walrasian exchange market.
We focus on the Ethereum blockchain for these reasons:

 \medskip
 (a) It has a mature exchange ecosystem, with relatively high liquidity and considerable value held in independent AMMs.

 (b) The block building process currently involves a complex combinatorial optimization task: selecting a compatible ordered subset of available transactions and bundles to maximize block value.
 There is no standard algorithm for solving this problem.
 Each block builder applies its chosen optimization methods to the order flow available to it.
 In practice, access to exclusive order flow is essential for a builder to remain competitive. By contrast, if all exchange operations participated in a Walrasian auction, the exchange-related part of block building would instead be organized around computing an equilibrium of the participating agents.

\medskip

An additional advantage of the Walrasian auction market structure is that trades are independent of their execution order, supporting efficient parallel execution \cite{Speedex}.
SPEEDEX is designed to run directly within blockchain node software, with orders
submitted through blockchain transactions. The authors report prototypes in a custom blockchain using HotStuff and in Stellar. For Ethereum, we prefer to explore an exchange implemented through smart contracts, with signed orders collected and batches built off-chain. This avoids requiring a separate on-chain transaction for each submitted order, including orders that remain unfilled, and allows the exchange to be deployed
without modifying Ethereum’s underlying protocol. Thus, in Section \ref{section:mechanism} we describe a plausible protocol design suitable for Ethereum’s application layer. For further context, we include a brief comparison with CoW Protocol at the end of Subsection \ref{subsubsec:basic-batch-structure}.

\subsection{Paper structure and by-products}

Our theoretical framework is based on exchange agents with continuous, single-valued response functions. This means that an agent is represented by a continuous function $\cc P_n \xto \varphi \R^n$, where $\cc P_n = (\R_{>0})^n$ is the space of prices of $n$ tokens and $\varphi(p)$ is the vector of amounts that the agent will exchange at prices $p$. A positive component $\varphi_i(p) > 0$ means that the agent sells $\varphi_i(p)$ units of the $i$-th token, while
$\varphi_i(p) < 0$ means that the agent buys $|\varphi_i(p)|$ units. Definition \ref{def:supply} introduces all the required properties for {\it supply functions} $\varphi$. For a finite set of agents $\cc A$, the sum $\Phi = \sum_{a \in \cc A} \varphi_a$ represents the aggregate supply function and satisfies Definition \ref{def:supply}. An equilibrium is represented by the expression
$$ \Phi(p) = 0 $$
Even though it is customary to consider multivalued functions (\cite{AD},  \cite{Augmenting}, for instance), and there is a rich theory for this more general case, the use of {\it supply functions} provides a simplified language and theoretical framework, with some distinctive results, such as the uniqueness property \ref{thm:net-trades}, with valuable practical consequences. At the same time, it is expressive enough to approximate usual limit orders with continuous limit orders (see Remark \ref{remark:approximate-limit-order}) and include general AMM responses, unifying the different types of agents. We expect this language to be directly useful in applications. Relevant references for single-valued response functions are \cite{Scarf1960} and \cite{Neuefeind1980}.

A pairwise supply function describes an agent that trades only two tokens, with its response depending only on their relative price. After the definition of supply functions, in Section \ref{section:supply-functions}, we give a complete characterization of pairwise supply functions in terms of a function of a single variable $\R_{>0} \xto g \R$, and present some interesting particular
examples.
Section \ref{section:existence} presents an equilibrium existence theorem (Theorem \ref{thm:existence}) that is well suited to the purpose of the main theorem, and relates it to a more general theorem \cite[Lemma 1]{Neuefeind1980}.
Throughout the text, we develop the required theory of exchange markets in a self-contained manner,
so that readers do not need to trace proofs through the literature, where definitions and conventions vary.
Section \ref{section:uniqueness} provides several results regarding uniqueness properties and
the structure of the set of equilibrium price vectors.
Theorem \ref{thm:net-trades} establishes the uniqueness of net trades at equilibrium, under a monotonicity hypothesis. This important result is not new (see \cite[Theorem 3.1]{PolterovichSpivak1984}) but our presentation here is particularly neat.

Section \ref{section:external-amms} introduces the formalism for external AMMs. For any given external AMM we construct two associated virtual agents. The one satisfying the monotonicity condition is selected to participate in the proof of Theorem \ref{thm:main}. Section \ref{section:convex-program} formulates equilibrium computation as a convex optimization problem, following ideas from \cite{Devanur2016} and \cite{Augmenting}. The construction is simple because trade orders and virtual agents share the same representation.

Appendix \ref{app:non-differentiable} generalizes the results of Sections \ref{section:external-amms} and \ref{section:batch-exchange} by allowing external AMMs with non-differentiable swap functions, i.e. with price ranges containing no liquidity.
Appendix \ref{app:decomposition} contains a useful result for describing equilibria when the market can be decomposed into two clusters of tokens with just a single token in common.

\begin{notation} $ $

\noindent $\sbt$ $n$ denotes a natural number, $n \ge 2$.

\noindent $\sbt$ The set of indices is $I_n = \{i \in \N / 1 \le i \le n\}$.

\noindent $\sbt$ The set of prices is $\PP_n = (\R_{>0})^n$.
\end{notation}

\section{Supply functions} \label{section:supply-functions}

In our formal framework, agents are represented by continuous {\it supply functions}.
They specify the amounts of tokens to trade depending on the price vector $p \in \PP_n$.
Therefore, the information associated with an agent is a function $\PP_n \xto \varphi \R^n$.
The amount $\varphi_i(p)$ indicates the amount of token $i$ supplied by the agent if the execution price is $p$. This means that $|\varphi_i(p)|$ is a selling amount if $\varphi_i(p) > 0$ and a buying amount if $\varphi_i(p) < 0$.

\begin{Definition}[Supply function] \label{def:supply} $ $
 We call $\PP_n \xlongto{\varphi} \R^n$ a supply function if

(0) $\varphi(\lambda p) = \varphi(p)$ $\forall \lambda > 0$, $\forall p \in \PP_n$, and $\varphi$ is continuous.

(1) $\varphi(p) \cdot p = 0 \ \forall p \in \PP_n$.

(2) $\varphi_i$ is bounded above $\forall i$.

(3) The functions $\psi_i(p) = \varphi_i(p)p_i$ extend continuously to
$(\R_{\ge 0})^n \xto{\psi_i} \R$.
\end{Definition}

{\bf Explanation and discussion of the conditions}

(0) This expresses the fact that a global factor $\lambda$ is meaningless, since all the relative prices remain untouched.

(1) This condition is called Walras's law and expresses the fact that agents only trade according to the price vector $p$.
The value of tokens sold equals the value of tokens bought.

Observe that it can be written in terms of $\psi$ as $\sum_{i=1}^n \psi_i = 0$. While $\varphi_i$ expresses the supply amount of token $i$, $\psi_i$ expresses the supply value.

(2) Agents can only supply a limited amount of each token.

(3) A boundary regularity assumption on net trade values. It allows demanded
quantities to diverge as prices vanish, while requiring the value vector to have a finite limit at every boundary price vector. The condition is easy to verify in our intended application and allows a simple proof of existence of equilibrium.

\medskip

\begin{Proposition} \label{linearity}
The sum of two supply functions is a supply function.
\end{Proposition}

\begin{proof}
 It is straightforward.
\end{proof}

\begin{Proposition} \label{neg-psi}
For any supply function $\varphi$, following the notation of Definition \ref{def:supply}, we have \medskip

\noindent (a) $\psi$ is homogeneous of degree 1: $\psi(\lambda p) = \lambda \psi(p)$ $\forall \lambda > 0$, $\forall p \in \PP_n$. In particular, $\psi(0) = 0$.

\noindent (b) $\psi_i(p) \le 0$ if $p_i = 0$ for any index $i \in I_n$.

\end{Proposition}

\begin{proof}
(a) Follows from the definition and homogeneity of $\varphi$.

(b) By condition (2), $\varphi_i \le M$ for some constant $M$. Therefore, $\varphi_i(q) q_i \le M q_i$.
  Taking the limit $q \to p$ we obtain $\psi_i(p) \le 0$.
\end{proof}

\subsection{Pairwise supply functions}

Even though most of the theory allows agents to trade more than two tokens, equilibrium computation might become significantly more complicated in this general case. A pairwise supply function describes an agent that trades only two tokens, with its response depending only on their relative price. In applications, it seems sufficient to restrict attention to this type of agent. The global supply function $\varphi$ will arise as a sum of pairwise supply functions. For this reason, we focus on the case $n=2$. We will fully characterize pairwise supply functions and exhibit the case of piecewise linear supply functions as an interesting example.

  \begin{Proposition} \label{prop:characterization}
Let $\cc P_2 \xto \varphi \R^2$ be a supply function. Then there is a
continuous function \mbox{$\R_{>0} \xto g \R$} satisfying
\begin{equation} \label{eq:def-varphi}
    \begin{cases}
      \varphi_1(p) = g(p_2/p_1)\\
      \varphi_2(p) = -\frac{p_1}{p_2}g(p_2/p_1)
    \end{cases}
\end{equation}
such that

(a) $\lim_{r \to 0} g(r) \in \R_{\ge 0}$ and $\lim_{r \to \infty} g(r)/r \in \R_{\le 0}$.

(b) $g$ is bounded above and $g(r)/r$ is bounded below.

\noindent Conversely, given any continuous function $\R_{>0} \xto g \R$ with properties (a) and (b), equation (\ref{eq:def-varphi}) defines a supply function for two tokens.
\end{Proposition}

\begin{proof} Given $\varphi$, define $\R_{>0} \xto g \R$ by $g(r) := \varphi_1(1,r)$.
 With this definition, $g$ is continuous. We can verify
 $$\varphi_1(p_1, p_2) = \varphi_1(1, p_2 / p_1) = g(p_2/p_1)$$
 and the expression for $\varphi_2$ follows by applying Walras's law.

To prove (a) we use the continuity of $\psi$ as follows.
 Taking $p_1 > 0$ we have $g(r) = \varphi_1(p_1, rp_1) = \frac{1}{p_1} \psi_1(p_1, rp_1)$, so that

 $$\lim_{r \to 0} g(r) = \frac{\psi_1(p_1, 0)}{p_1} = \frac{-\psi_2(p_1, 0)}{p_1} \ge 0$$
 We have used $\psi_1 + \psi_2 = 0$ and Proposition \ref{neg-psi}.

 Taking $p_2 > 0$, we have
$\frac{1}{r} g(r) = \frac{1}{r} \varphi_1(p_2/r, p_2) = \frac{1}{p_2} \psi_1(p_2/r,p_2)$
 so that
 $$\lim_{r \to \infty} g(r) / r = \frac{\psi_1(0, p_2)}{p_2} \le 0$$
 $g$ is bounded above because $\varphi_1$ is, and
 $g(r)/r = -\varphi_2(1, r)$ is bounded below because
 $\varphi_2$ is bounded above.

 For the converse, assume that $\R_{>0} \xto g \R$ is continuous
 and satisfies (a) and (b). Conditions (0) and (1) from Definition \ref{def:supply}
are straightforward. Condition (2) follows from (b). For condition (3), we write explicitly, for $p_1, p_2 > 0$:
$$\psi_1(p) = p_1 \varphi_1(p) = p_1 g(p_2 / p_1)$$
$$\psi_1(p_1, 0) = p_1 \ell$$
$$\psi_1(0, p_2) = p_2 \ell'$$
$$\psi_1(0, 0) = 0, \ \ \psi_2 = -\psi_1$$
where $\ell = \lim_{r \to 0} g(r)$, and $\ell' = \lim_{r \to \infty} g(r)/r$.
These define the required continuous mapping
$$(\R_{\ge 0})^2 \xto \psi \R^2$$
\end{proof}

There is a symmetry between the indices $1$ and $2$ that is not explicit in the previous proposition.
\begin{Definition} \label{def:conjugate}
 \normalfont For a continuous function $\R_{>0} \xto g \R$ satisfying (a) and (b) from Proposition \ref{prop:characterization}, we denote by $\tilde g$ the function $\tilde g(r) = -rg(1/r)$
\end{Definition}

The function $\tilde g$ is such that $\varphi_2(p) = \tilde g(p_1 / p_2)$, thus inverting the roles of the indices. Since the notion of supply function is symmetric, $\tilde g$ automatically satisfies (a) and (b).

We also have $\tilde {\tilde g} = g$.

\begin{Examples} \label{examples} $ $ \normalfont

Here we show useful examples of functions $g$ that satisfy the conditions of Proposition \ref{prop:characterization} and implicitly consider their associated supply functions $\varphi$.

 (I) Linear functions $g_{a,b}(r) = -ar + b$, with $a, b \ge 0$ . Notice that
 $\widetilde {g_{a,b}} = g_{b, a}$.

 (II) More generally, we can take any continuous function $\R_{>0} \xto g \R$ that:

\noindent coincides with $g_{a,b}$ in a neighborhood of $0$, with $a \in \R$, $b \ge 0$, and

\noindent coincides with $g_{c,d}$ in a neighborhood of $\infty$, with $c \ge 0$, $d \in \R$.

(III) In particular, piecewise linear functions $g$ satisfying (II).
Notice that $\tilde g$ is also piecewise linear and satisfies (II), since it can be constructed by inverting each piece.

(IV) As a particular case of (III) we can construct continuous limit sell orders with parameters $(b, r_1, r_2)$ as follows:
\[
 g_{\text{sell}}(r) =
\begin{cases}
  b & \mbox{if } r \le r_1 \\
  \mbox{linear interpolation} & \mbox{if } r_1 \le r \le r_2 \\
  0 & \mbox{if } r \ge r_2
\end{cases}
\]
where $b > 0$, $0 < r_1 < r_2$.
Similarly, we define continuous limit buy orders with parameters
$(d, r_1, r_2)$ as follows:
\[
 g_{\text{buy}} (r) =
\begin{cases}
  0 & \mbox{if } r \le r_1 \\
  \mbox{linear interpolation} & \mbox{if } r_1 \le r \le r_2 \\
  -d & \mbox{if } r \ge r_2
\end{cases}
\]
where $d > 0$, $0 < r_1 < r_2$.

Note that $g_{\text{sell}}$, $g_{\text{buy}}$ and $\widetilde{g_{\text{sell}}}$ are decreasing but $\widetilde{g_{\text{buy}}}$ is not monotone.

(V) Any continuous function $\R_{\ge 0} \xto g \R_{\ge 0}$ such that
$\underset{\infty}{\lim} \ g \in \R_{\ge 0}$.

\end{Examples}

\begin{Remark} \label{remark:approximate-limit-order}
  \normalfont
  Consider a continuous limit sell order with parameters $(b, r_1, r_2)$ with $r_1$ and $r_2$ very close to each other, say $r_2 / r_1 \le 1 + 10^{-6}$. The corresponding supply function closely approximates a traditional limit sell order that allows partial execution.
  Continuous supply functions cannot approximate orders that do not allow partial execution.
\end{Remark}

\begin{Examples} \label{ex:CP-AMM}
\normalfont
Consider the {\it function maximizing AMM} of Canidio and Fritsch \cite[\S 2]{FMAMM}, in the case where the objective is the product of its reserves and no swap fee is charged. If its initial reserves are $x,y>0$ of tokens $1$ and $2$, respectively, then at prices $p\in\PP_2$ it chooses new reserves $u,v\geq0$ to maximize $uv$, subject to
\[ p_1 (x - u) +p_2 (y - v) =0 \]
The maximum is attained when the two reserves have equal value, giving
\[
 u=\frac{x}{2}+\frac{p_2}{p_1}\frac{y}{2}
 \qquad
 v=\frac{p_1}{p_2}\frac{x}{2}+\frac{y}{2}
\]
Since the amounts supplied by the AMM are $x-u$ and $y-v$, its supply function is associated, in the sense of Proposition~\ref{prop:characterization}, to the linear function $g_{y/2, x/2}$ defined in Examples \ref{examples} (I).
\end{Examples}

\subsection{Supply functions supported on a set of tokens}

  If we have a supply function $\PP_2 \xto \varphi \R^2$, we can extend it to a supply function $\PP_n \xto \varphi \R^n$ by carrying the indices $(1, 2)$ to any pair of indices $(i, j)$ and ignoring the other variables. This should be intuitively clear, but it is useful to state some precise facts about this procedure and its reverse, the restriction. The following definition is also used in Proposition \ref{prop:decomposition}.

\begin{Definition}
  Let $\PP_n \xto \varphi \R^n$ be a supply function, and
  $I \subset I_n$. We say that $\varphi$ is supported on $I$ if
  $\varphi(p)$ does not depend on $p_j$ for every $j \notin I$.
\end{Definition}

\begin{Proposition} \label{prop:support}
 A supply function $\PP_n \xto \varphi \R^n$ supported
 on $I \subset I_n$ satisfies $\varphi_j(p) = 0$ for every $j \notin I$.

 Moreover, if $I = \{i_1,...,i_m\}$, the restricted function
 $\PP_m \xto {\tilde \varphi} \R^m$, defined by
 $\tilde \varphi(p) = \pi \varphi(\hat p)$, where $\hat p$ is any vector
 such that $\hat p_{i_j} = p_j$, and $\pi(x)_j = x_{i_j}$, is a supply function.
\end{Proposition}

\begin{proof}
 Let $j \notin I$. Take $p, q \in \PP_n$ that only differ at the $j$-th coordinate.
 By hypothesis we have $\varphi(p) = \varphi(q)$.
 Then we have
 $$p \cdot \varphi(p) = q \cdot \varphi(p) = 0$$
 $$p_j \varphi_j(p) = q_j \varphi_j(p)$$
 $$\varphi_j(p) = 0$$
 The proof of the second assertion is just a sequence of trivial checks.
\end{proof}

Conversely, we can extend supply functions by adding trivial variables.

\begin{Lemma} \label{lemma:extend}
  Let $m \leq n$, let $\PP_m \xto \varphi \R^m$ be a supply
  function and let $I_m \stackrel{\sigma}{\hookrightarrow} I_n$ be injective.
Write $p_\sigma = (p_{\sigma(1)},...,p_{\sigma(m)})$.
Then the function $\PP_n \xto {\bar \varphi} \R^n$ defined by
$\bar \varphi_{\sigma(j)}(p) = \varphi_j(p_\sigma)$ and $\bar \varphi_i(p) = 0$ if $i \notin \sigma(I_m)$ is a supply function supported on $\sigma(I_m)$.
\end{Lemma}

\begin{Definition}
 A pairwise supply function is a supply function supported on a set of two indices. Equivalently (by Proposition \ref{prop:support}), it is a supply function obtained as in Lemma \ref{lemma:extend} with $m=2$.
\end{Definition}

\section{Existence of equilibrium} \label{section:existence}

In this section, we apply Brouwer's fixed point theorem to obtain sufficient conditions for the existence of equilibrium prices (Theorem \ref{thm:existence}).

  \begin{Definition}
 Let $1 \le i,j \le n$ be two indices. We say that a supply function $\varphi$ connects $i \to j$ if $\psi_i(p) < 0$ whenever $p_i = 0$ and $p_j > 0$,
 where $\psi_i$ are the functions introduced in Definition \ref{def:supply}.
\end{Definition}

The condition ``$\varphi$ connects $i \to j$'' means that the demanded value for token $i$ at prices $q$ close to $p$, with $p_i = 0$, $p_j > 0$, is uniformly positive.

\begin{Definition}[Strict supply function] \label{def:strict}
 We say that a supply function $\PP_n \xto \varphi \R^n$ is strict if it  ``connects the indices''. Explicitly, for every ordered pair of indices
 $(j_0,j_1)$ there is a sequence $i_0, ..., i_m$ of indices such that
 $j_0 = i_0$, $j_1 = i_m$ and $\varphi$ connects $i_k \to i_{k+1}$
 $\forall k \in \{0, ..., m - 1\}$ (i.e. the corresponding directed graph
 is strongly connected).
\end{Definition}

\begin{Lemma} \label{boundary}
  Let $\PP_n \xlongto \varphi \R^n$ be a strict supply function.
  For every $p \in \R_{\ge 0}^n \setminus \PP_n$, $p \neq 0$, there is an index $k$ such that $\psi_k(p) < 0$.
\end{Lemma}

\begin{proof}
  Let $I$ be the set of indices $i$ such that $p_i = 0$ and $J$ the complement.
$J$ is non-empty because $p \neq 0$, and $I$ is non-empty because $p \notin \PP_n$.
Since by hypothesis $\varphi$ connects all the indices, there are indices $i \in I$ and $j \in J$ such that $\varphi$ connects $i \to j$. Therefore, $\psi_i(p) < 0$.
\end{proof}

\begin{Theorem}[Existence of equilibrium] \label{thm:existence}
Let $\PP_n \xlongto \varphi \R^n$ be a strict supply function.
Then there is a $p \in \PP_n$ such that $\varphi(p) = 0$.
\end{Theorem}

\begin{proof}
Let
$$\Delta = \left\{ q \in \R^n / q_i \ge 0, \ \sum_{i=1}^n q_i = 1 \right\}$$ be the $(n-1)$-simplex, and $\stackrel{\circ}{\Delta} = \{q \in \Delta / q_i > 0 \ \forall i\}$. Consider the function $\stackrel{\circ}{\Delta} \xto \rho \R^n$ defined by
$$\rho_i(q) = q_i \left(1 - \frac{1}{M_i}\varphi_i (q/M) \right)$$
where for each index $i$, $M_i > 0$ is an upper bound for $\varphi_i$, and
$q/M$ denotes the vector whose $i$-th coordinate is $q_i / M_i$.

The function $\rho$ has the following properties:

\noindent (a) $\rho_i(q) \ge 0$ $\forall q \in \stackrel{\circ}{\Delta}$.

\noindent (b) $\sum_{i=1}^n \rho_i(q) = 1$.

\noindent (c) $\rho_i(q)$ extends continuously to $\Delta$.

Property (a) holds because $M_i$ is an upper bound for $\varphi_i$;
property (b) follows from Walras's law; and property (c) follows from the extension property (3) of Definition \ref{def:supply}.

From properties (a), (b), (c) we conclude that $\rho$ is a continuous function
$$\Delta \xlongto \rho \Delta$$
By Brouwer's fixed-point theorem, we have a $\tilde q \in \Delta$ such that
$\rho(\tilde q) = \tilde q$. Set $p = \tilde q / M$. After simplification,
we have $\psi(p) = 0$. By the previous Lemma \ref{boundary}, $p$ cannot be at the boundary. Therefore $p \in \PP_n$ and $\varphi(p)=0$.
\end{proof}

\noindent {\bf Remark.} For $n = 2$, an example of a supply function with no roots is
  $$\varphi(p_1, p_2) = \left(-\frac{p_2}{p_1}, 1 \right)$$
In this case $\psi(1, 0) = (0, 0)$.

\begin{Proposition} \label{prop:sum-connects}
Let $\varphi_1, \varphi_2$ be supply functions. If $\varphi_1$ connects $i \to j$, then $\varphi_1 + \varphi_2$ also connects $i \to j$.
\end{Proposition}

The proof makes use of Proposition \ref{neg-psi} (b).

\begin{Proposition} \label{prop:connects-indices}
For $n=2$, using the notation of Proposition \ref{prop:characterization},

$\sbt$ $\varphi$ connects $2 \to 1$ if and only if $\lim_{r \to 0} g(r) > 0$,

$\sbt$ $\varphi$ connects $1 \to 2$ if and only if $\lim_{r \to 0} \tilde g(r) = -\lim_{r \to \infty} g(r)/r > 0$,
\end{Proposition}

\medskip
\noindent {\bf Relation to Neuefeind's existence lemma}
\medskip

  Theorem \ref{thm:existence} is a special case of Neuefeind's lemma
  \cite[Lemma 1]{Neuefeind1980}. For a strict supply function $\varphi$, the boundary condition of \cite[Lemma 1]{Neuefeind1980} holds by taking
  $$D = \overset{\circ}{\Delta} \text{ and } \bar p = \frac{1}{n} (1, ..., 1)$$

  Strictness therefore provides a sufficient condition that is easy to verify in practice for a sum of pairwise supply functions. In this context, the proof of Theorem \ref{thm:existence} is more direct than the proof through \cite[Lemma 1]{Neuefeind1980}.

\section{Uniqueness properties of equilibrium} \label{section:uniqueness}

Under a monotonicity property called {\it weak gross substitutability} (WGS), powerful uniqueness properties can be proved. Most notably, Theorem \ref{thm:net-trades} states that two different equilibrium vectors lead to the same effective trades. Also under WGS, Theorem \ref{thm:ratio} clarifies the structure of the set of equilibrium vectors $\E$. As a corollary, the set $\E$ is convex, and so is the set $\log \E$ of coordinatewise logarithms.

For the application to an exchange protocol, this means that given a set of agents, the effect of trading at the equilibrium is completely determined, while a convex set describes the valid equilibrium price vectors.

Finally, we include the analog of \cite[Theorem 4]{Speedex} in this language.

The WGS property can be traced back to the gross-substitutes condition used by Wald in his 1936 study of equilibrium uniqueness~\cite{Wald1936}. For more background on WGS in exchange economies, see Codenotti et al.~\cite{Codenotti}, who build on the classical convexity of the equilibrium price set to develop
algorithms for computing approximate equilibria. These concepts are also present in \cite{Speedex} and \cite{Augmenting}.

\begin{Definition} \label{def:wgs}
A supply function $\PP_n \xto \varphi \R^n$ satisfies weak gross substitutability (WGS) iff
\begin{equation}\label{eq:wgs}
 p\le q,\quad p_i=q_i
 \quad\Longrightarrow\quad
 \varphi_i(q)\le\varphi_i(p)
 \qquad(i=1,\ldots,n)
\end{equation}
The vector inequality $p\le q$ is coordinatewise.
\end{Definition}

We can phrase the meaning of WGS as follows:
if, at prices $q$, each unit of token $i$ buys no more of any other token than it does at prices $p$, then the net amount of token $i$ offered for sale is no greater at $q$ than at $p$.

We anticipate that continuous limit sell orders satisfy WGS, while continuous limit buy orders do not (see Proposition \ref{prop:pairwise-WGS} below). The WGS property is stable under addition.

\medskip
\noindent {\bf Notation:} For a supply function $\varphi$, we denote by
$$\E_\varphi = \{p \in \PP_n : \varphi(p) = 0\}$$
the set of equilibrium price vectors.

The results of this section depend on the following simple lemma. This {\it lattice property} has been treated in a more general setting in \cite[Corollary~1]{PolterovichSpivak1983} and \cite[Corollary~1 and Section~3.3]{GalichonSamuelsonVernet2025}.

\begin{Lemma}[Lattice property] \label{lemma:lattice}
Let $\varphi$ be a supply function satisfying WGS. If $p,q \in \E_\varphi$, their coordinatewise minimum $m=p \wedge q$, defined by $m_j=\min\{p_j,q_j\}$, also belongs to $\E_\varphi$. Analogously, $p \vee q \in \E_\varphi$.
\end{Lemma}

\begin{proof}
For a coordinate $j$, choose $w$ equal to $p$ or $q$ such that
$m_j=w_j$. Because $m \le w$, WGS gives
\[
 \varphi_j(m)\ge\varphi_j(w)=0
\]
Since $m>0$ coordinatewise and $m\cdot\varphi(m)=0$, all coordinates of $\varphi(m)$ vanish.
\end{proof}

Neither the previous lemma nor the next theorem actually requires continuity of the supply function or conditions (2) and (3) from Definition \ref{def:supply}.

\begin{Theorem}[Unique equilibrium net trades]\label{thm:net-trades}
Let $\{\varphi^a\}_{a \in \cc A}$ be a finite family of supply functions satisfying WGS and
$$\Phi = \sum_{a \in \cc A} \varphi^a$$
If $p,q \in \E_\Phi$, then
\[
 \varphi^a(p)=\varphi^a(q)\qquad\text{for every }a\in\mathcal A
\]
\end{Theorem}

Thus the entire profile of equilibrium net trades is unique, although
equilibrium prices need not be unique up to scale.

\begin{proof}
Fix a token $i$. Scale $q$ to agree with $p$ in that coordinate:
\[
 \lambda=\frac{p_i}{q_i}>0,\qquad
 \widehat q=\lambda q,\qquad
 m=p\wedge\widehat q
\]
Homogeneity makes $\widehat q$ an equilibrium, and Lemma~\ref{lemma:lattice}
makes $m$ an equilibrium. Moreover,
\[
 m_i=p_i=\widehat q_i,\qquad m\le p,\qquad m\le\widehat q
\]
Apply WGS to each agent separately:
\[
 \varphi_i^a(m)-\varphi_i^a(p)\ge0,\qquad
 \varphi_i^a(m)-\varphi_i^a(\widehat q) \ge 0
\]
Market clearing gives
\[
 \sum_{a\in\mathcal A}
 \bigl(\varphi_i^a(m)-\varphi_i^a(p)\bigr)
 =\Phi_i(m)-\Phi_i(p)=0
\]
Every summand is nonnegative, so every summand is zero. Applying the
same reasoning with $\widehat q$ in place of $p$ yields
\[
 \varphi_i^a(p)=\varphi_i^a(m)=\varphi_i^a(\widehat q)
 =\varphi_i^a(q)
\]
where the final equality uses homogeneity. Since $i$ was
arbitrary, this proves equality of the whole supply vector for every
agent $a \in \cc A$.
\end{proof}

In the following theorem, the strictness hypothesis might be removed, but in that case the constants $U_{ij}$ could be infinite.

\begin{Theorem}[Equilibrium set characterization]\label{thm:ratio}
Let $\PP_n \xto \varphi \R^n$ be a strict supply function satisfying WGS. For $i,j \in I_n$ define
\[
 U_{ij}=\max_{q \in \E_\varphi}\frac{q_i}{q_j}<\infty
\]
Then
\[
 \E_\varphi=\{p>0 : \ \frac{p_i}{p_j} \le U_{ij} \ \forall i,j \in I_n\}
\]
Equivalently, $\log \E_\varphi$ is the polyhedron
\[
 \log\E_\varphi=\{x\in\R^n : \ x_i-x_j\le\log U_{ij} \ \forall i,j \in I_n\}
\]
\end{Theorem}

\begin{proof}
Simplify the notation by writing $\E = \E_\varphi$.
Recall the notation $\Delta = \{q\ge0:\sum_i q_i=1\}$ from the proof of Theorem \ref{thm:existence}.
The set $\E_\Delta =\E \cap \Delta$ is compact and contained in
the interior of $\Delta$.
The maximum defining $U_{ij}$ exists and it is finite because the ratio is continuous, bounded, and scale invariant. It is attained on the compact normalized set $\E_\Delta$. The inclusion $\E \subset
\{p : \frac{p_i}{p_j} \le U_{ij}\}$ is immediate.

Conversely, suppose that $p>0$ satisfies all these inequalities. For each
ordered pair $(i,j)$ choose $q^{ij}\in\E$ attaining $U_{ij}$ and such that $q^{ij}_i=p_i$. We have:
$$\frac{p_i}{p_j} \le U_{ij} = \frac{q_i^{ij}}{q_j^{ij}} = \frac{p_i}{q_j^{ij}}$$
implying $q^{ij}_j \le p_j$. For each $i$ let
$m^i=\bigwedge_j q^{ij}$. Then $m^i\in\E$ by the lattice Lemma \ref{lemma:lattice}, $m^i_i=p_i$, and $m^i_j\le p_j$ for every $j$. Finally set $w=\bigvee_i m^i \in \E$. We have $w\le p$ and, for each coordinate $j$, $w_j\ge m^j_j=p_j$. Hence $w=p$, proving $p\in\E$.
\end{proof}

\begin{Corollary} \label{cor:convexity}
Let $\PP_n \xto \varphi \R^n$ be a strict WGS supply function.
Both $\E_\varphi$ and $\log \E_\varphi$ are convex.
\end{Corollary}

\begin{proof}
By Theorem~\ref{thm:ratio},
\[
 \E_\varphi=\bigcap_{i,j}\{p\in\R^n_{>0}:p_i-U_{ij}p_j\le0\}
\]
\[
 \log\E_\varphi=\bigcap_{i,j}\{x\in\R^n:x_i-x_j\le\log U_{ij}\}
\]
Each set in each intersection is convex.
\end{proof}

\subsection{WGS and pairwise supply functions}

We want to characterize pairwise supply functions that satisfy WGS.

\begin{Proposition} \label{prop:pairwise-WGS}
  A supply function $\PP_2 \xto \varphi \R^2$ satisfies WGS if and only if both the function $g$ from Proposition \ref{prop:characterization} and its conjugate $\tilde g$ (Definition \ref{def:conjugate}) are decreasing.
\end{Proposition}

\begin{proof}
By definition, $\varphi_1(1, r) = g(r)$ and $\varphi_2(r, 1) = \tilde g(r)$

\noindent For $r < s$,
$$\varphi_1(1, r) \ge \varphi_1(1, s) \Longleftrightarrow g(r) \ge g(s)$$
$$\varphi_2(r, 1) \ge \varphi_2(s, 1) \Longleftrightarrow \tilde g(r) \ge \tilde g(s)$$
By homogeneity, this covers all cases.
\end{proof}

Under the decrease condition, boundedness from above is automatic for $g$ and $\tilde g$.
This motivates the following definition.

\begin{Definition} \label{def:WGS-function}
 A function $\R_{>0} \xto g \R$ is a WGS function if it is continuous and decreasing, $s \mapsto \frac{g(s)}{s}$ is decreasing,
 $$\lim_{s \to 0} g(s) \in \R_{\ge 0} \mbox{, and } \lim_{s \to \infty} \frac{g(s)}{s} \in \R_{\le 0}$$

\end{Definition}

\subsection{Further uniqueness properties}

The following proposition and corollary express
\cite[Theorem 4 and Corollary 1]{Speedex}, respectively, in the present framework.

\begin{Proposition}
Let $\{\varphi^a\}_{a \in \cc A}$ be a finite family of pairwise supply functions satisfying WGS.
Consider the aggregate
$$\Phi = \sum_{a \in \cc A} \varphi^a$$
Let $p,q \in \E_\Phi$. If two tokens $i, j$ satisfy $\frac{p_i}{p_j} \neq \frac{q_i}{q_j}$,
then there is a partition $\{1,\ldots,n\}=S\mathbin{\dot\cup}T$ with $i \in S$ and
$j \in T$ such that $\varphi^a(p)=\varphi^a(q)=0$ for every agent whose support meets both sets.
Thus neither equilibrium contains any nonzero pairwise trade across the
partition.
\end{Proposition}

\begin{proof}
Set $\lambda_k=p_k/q_k$. Without loss of generality, we may assume
$\lambda_i<\lambda_j$. Choose $t$ strictly between these two values and
not equal to any $\lambda_k$, and define
\[
 S=\{k:\lambda_k<t\} \qquad
 T=\{k:\lambda_k>t\}
\]
Then $i \in S$, $j \in T$, and for every $l \in S$, $m \in T$,
\[
 \frac{p_l}{p_m}<\frac{q_l}{q_m}
\]
\noindent By Theorem \ref{thm:net-trades}, for every agent $a$,
\[
 \varphi^a(p)=\varphi^a(q)=:z^a
\]
Consider an agent $b$ supported on $\{l,m\}$ with $l\in S$ and
$m\in T$. Applying Walras's law at both equilibria gives
\[
  \begin{pmatrix} p_l & p_m \\ q_l & q_m \end{pmatrix}
  \begin{pmatrix} z_l^b \\ z_m^b \end{pmatrix}
  =
  \begin{pmatrix} 0 \\ 0 \end{pmatrix}
\]
Since the determinant is nonzero, $z^b=0$.
\end{proof}

\begin{Corollary}
Under the hypotheses of the proposition, fix $p\in\E_\Phi$ and form an undirected graph with one vertex per token, joining $i$ and $j$ whenever some agent supported on $\{i,j\}$ has $\varphi^a(p) \ne 0$. If this graph is connected, then all equilibrium price vectors are proportional to $p$.
\end{Corollary}

\begin{proof}
 By the connectedness condition, there is no nontrivial partition as in the proposition statement, forcing every other equilibrium vector to be proportional to $p$.
\end{proof}

\section{External AMMs and their associated supply functions} \label{section:external-amms}

External AMMs are an essential ingredient of the main theorem
(Theorem  \ref{thm:main}).
Since this theorem was intended for application on live blockchains, it is important to work with enough generality, including as many relevant AMMs as possible. Throughout this section, we impose a simplifying differentiability assumption, condition \ref{def:ammf} (d), to make the exposition cleaner. We postpone the details of the generalization to Appendix \ref{app:non-differentiable}.

Consider an AMM for two tokens $i$ and $j$ in a given state. We can swap $i$ for $j$ or $j$ for $i$. If we restrict attention to swapping $i$ for $j$,
we can model the AMM as a function $\R_{\ge 0} \xto f \R_{\ge 0}$.
For every incoming amount $x$ of token $i$, the AMM returns an amount
$f(x)$ of token $j$.
The function $f$ accounts for the pool fee, so that it reflects the actual behavior of the AMM.
For example, a constant-product AMM with reserves $a$ and $b$ and no pool fees has $f(x) = bx / (a + x)$.

\begin{Definition} \label{def:ammf}
 An AMM function is a map $\R_{\ge 0} \xto f \R_{\ge 0}$ with the following properties:

 (a) It is increasing.

 (b) $f(0) = 0$, but $f \not\equiv 0$.

 (c) It is bounded above.

 (d) It is continuously differentiable, with $f'$ strictly decreasing.
\end{Definition}

\noindent Differentiability at $0$ is interpreted as right differentiability.

\medskip
Representing an AMM by two continuously differentiable functions, one for each trade direction, covers the case of standard AMMs defined by differentiable curves even when the swap fee causes a loss of differentiability at the current state.
Moreover, this definition includes Uniswap's concentrated liquidity pools as long as there is nonzero liquidity at every price. The restriction of nonzero liquidity everywhere can be lifted. This can be understood intuitively as a limiting case, or, as in Appendix \ref{app:non-differentiable}, by dealing with a more general definition of AMM function.

\medskip

\noindent For an AMM function $f$, we have the limits
$$\lim_{x \to \infty} f(x) = M = \sup\operatorname{Im}(f) < \infty$$
$$\lim_{x \to \infty} f'(x) = 0$$
Thus the derivative gives a bijection $[0, +\infty) \xto{f'} (0, f'(0)]$.

\begin{Definition} \label{def:h_f}
Given an AMM function $f$, consider $(f')^{-1}$ and a continuous extension $h_f$
$$(0, f'(0)] \xto{(f')^{-1}} [0, +\infty)$$
$$(0, +\infty) \xrightarrow{\hspace*{.3cm} h_f \hspace*{.3cm}} [0, +\infty)$$
defined by $h_f(r) = 0$ for $r \ge f'(0)$.

\end{Definition}

The following definition and proposition associate two possible supply functions to any given AMM function. In Theorem \ref{thm:main}, this construction is leveraged to produce virtual agents that participate in the exchange economy.

\begin{Definition} \label{def:g-and-G}
 For any AMM function $f$, define
 $$\R_{> 0} \xlongto {g_f} \R, \ \ \  g_f(r) = rh_f(r)$$
 $$\R_{> 0} \xlongto {G_f} \R, \ \ \  G_f(r) = f(h_f(r))$$
with $h_f$ the extension of $(f')^{-1}$ previously introduced.
\end{Definition}

\begin{Proposition} \label{prop:external-AMM} $ $
 Let $f$ be an AMM function, and $M = \sup\operatorname{Im}(f)$.

(1) Both functions are continuous and satisfy
$$\lim_{r \to 0} g_f(r) = 0, \ \ \lim_{r \to \infty} g_f(r)/r = 0$$
$$\lim_{r \to 0} G_f(r) = M, \ \ \lim_{r \to \infty} G_f(r)/r = 0$$

(2) $g_f \le G_f$. The strict inequality $g_f(r) < G_f(r)$ holds if and only if $r < f'(0)$.

(3) $g_f$ and $G_f$ satisfy the conditions (a) and (b) of Proposition \ref{prop:characterization}.

(4) $G_f$ is a WGS function, while $g_f$ is not.

(5) The supply function $\cc P_2 \to \R^2$ associated to $G_f$ connects the indices $2 \to 1$, but does not connect $1 \to 2$. The supply function associated to $g_f$ does not connect the indices.

\end{Proposition}

\begin{figure}[ht]
  \centering
  \includegraphics[width=0.6\textwidth]{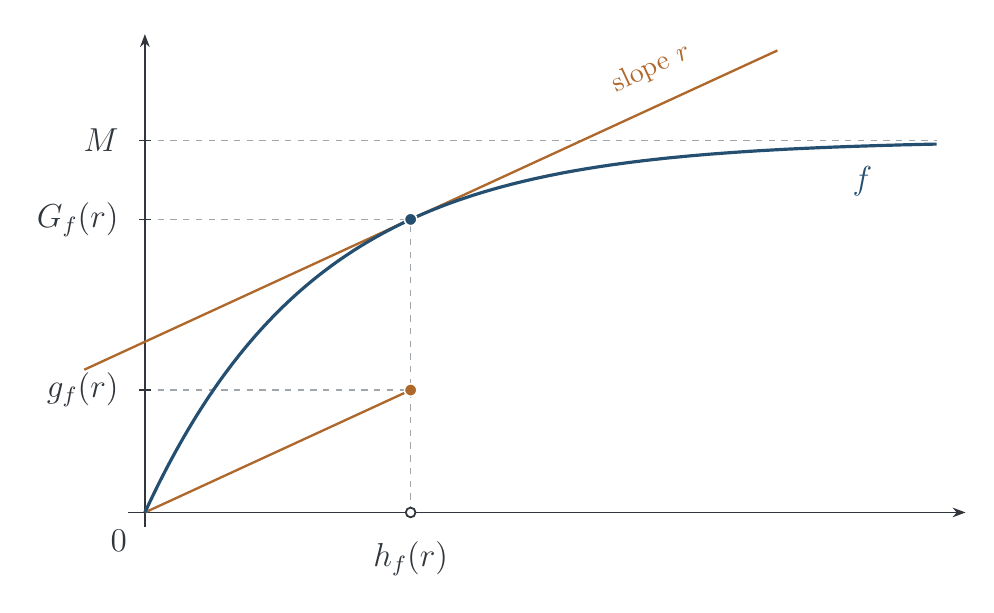}
  \caption{Representation of quantities associated to an external AMM function $f$}
  \label{figure:1}
\end{figure}

\begin{proof}
Since $f'$ is continuous and strictly decreasing, its inverse is continuous
and strictly decreasing on $(0,f'(0)]$, with $(f')^{-1}(f'(0))=0$.
Consequently, its extension $h_f$ is continuous and decreasing on
$\R_{>0}$, and $h_f(r)\to\infty$ as $r\to0$.

(1) Continuity of $h_f$ and $f$ implies continuity of $g_f$ and $G_f$.
Both functions vanish for $r\ge f'(0)$, which proves the two limits at
infinity. Also, since $f(x)\to M$ as $x\to\infty$, we have
$G_f(r)=f(h_f(r))\to M$ as $r\to0$.

To prove the remaining limit, observe that, for every $x>0$, the decrease
of $f'$ gives
\[
0\le \frac{x}{2}f'(x)
\le \int_{x/2}^{x} f'(t)\,dt
=f(x)-f(x/2)
\]
The right-hand side tends to zero as $x\to\infty$, so $xf'(x)\to0$. Now
$$\lim_{r \to 0} g_f(r) = \lim_{r \to 0} r f'^{-1}(r)
= \lim_{r \to 0} f'\left(f'^{-1}(r)\right) f'^{-1}(r) =
\lim_{x \to \infty} f'(x) x = 0 $$

(2) If $r \ge f'(0)$, then $h_f(r)=0$ and $g_f(r) = G_f(r) = 0$.
If $0 < r < f'(0)$, since $f'$ is strictly decreasing,
for every $x > 0$,
\[
0 < xf'(x) < \int_0^x f'(t)\,dt=f(x)
\]
Taking $x = h_f(r)$, we obtain $f'(x)=r$ and $0 < g_f(r) < G_f(r)$.

\medskip
(3) The limits in (1) give condition (a) of Proposition
\ref{prop:characterization}. For condition (b), we apply
\[ 0 \le g_f(r) \le G_f(r) \le M\]
to conclude that both functions are bounded above, and both $g_f(r)/r$ and $G_f(r)/r$ are bounded below by zero.

\medskip
(4) $G_f=f\circ h_f$ is decreasing because $h_f$ is decreasing and
$f$ is increasing. Since in addition $G_f \ge 0$, automatically $\frac{G_f(r)}{r}$ is also decreasing. On the other hand, $g_f(r) > 0$ for $0 < r < f'(0)$, whereas its limit at zero is zero. Thus $g_f$ cannot be decreasing.

\medskip
(5) This is immediate from Proposition \ref{prop:connects-indices}.
\end{proof}

We next introduce the formal structure for a system of external AMMs.
Following our definition of {\it AMM function} we adopt the convention that
AMMs operate only in one direction. In this way, a real-world external AMM for a pair of tokens in a given state is modeled as two AMMs in our formalism. In reality, these two AMMs cannot execute swaps of positive amounts simultaneously, but this does not create any conflict since that is also the case when applying Theorem \ref{thm:main}.

\begin{Definition}
 A system of external AMMs $\cc C = (C, \inn, \out , (f_c)_{c \in C})$ on $n$ tokens consists of a finite set $C$ (the set of AMMs), a pair of functions $C \xto{\operatorname{in}, \out } I_n$, with $\inn (c) \neq \out (c) \ \forall c$, and
 for each $c \in C$, an AMM function $f_c$ as in definition \ref{def:ammf}.

\end{Definition}

The maps ``$\inn$'' and ``$\out$'' assign to each AMM the indices of its incoming and outgoing tokens, respectively.

\begin{Definition}[Virtual agents and notation] \label{def:va-and-not}
\normalfont
 Let $\cc C = (C, \inn , \out , (f_c)_{c \in C})$ be a system of external AMMs on $n$ tokens.

 (a) Let
 $$h_c = h_{f_c}, \qquad g_c = g_{f_c}, \qquad G_c = G_{f_c}$$
 be the functions introduced in Definitions \ref{def:h_f} and \ref{def:g-and-G}, and define the {\it surplus}
 $$s_c(r) = \frac{f_c(h_c(r))}{r} - h_c(r) = \frac1{r}\Big(G_c(r)-g_c(r)\Big)\ge 0$$
 Its interpretation will become apparent soon. The inequality is just Proposition \ref{prop:external-AMM} (2).

 (b) For any $p \in \cc \PP_n$, we consider the vector
 $$h^{\cc C, p} \in \R_{\ge 0}^C, \ (h^{\cc C, p})_c = h_c(p_{\inn (c)}/p_{\out (c)})$$
 This is the vector of input amounts needed at each AMM $c$ so that the final marginal price is in accordance with $p$.

 (c) For each $c \in C$ we consider the supply function
\[
    \begin{cases}
      (\varphi^c)_{\inn (c)}(p) = -\frac{p_{ \out (c)}}{p_{\inn (c)}} G_c \left( \frac{p_{\inn (c)}}{p_{\out (c)}} \right) \\
      (\varphi^c)_{ \out (c)}(p) = G_c \left( \frac{p_{\inn (c)}}{p_{\out (c)}} \right) \\
      (\varphi^c)_a(p) = 0 & \mbox{if } a \neq \inn (c), \out (c)
    \end{cases}
\]

 (d) We associate to the system $\cc C$ the supply function
 $$\varphi^{\cc C} = \sum_{c \in C} \varphi^c$$
\end{Definition}

The functions defined in (c) are valid supply functions by Propositions \ref{prop:characterization} and \ref{prop:external-AMM}, together
with an almost trivial application of Lemma \ref{lemma:extend}.

The function $\varphi^c$ can be interpreted as a virtual agent associated with $c$, and $s_c$ as the surplus of the actual swap relative to that virtual agent.

\begin{figure}[ht]
  \centering
  \includegraphics[width=0.6\textwidth]{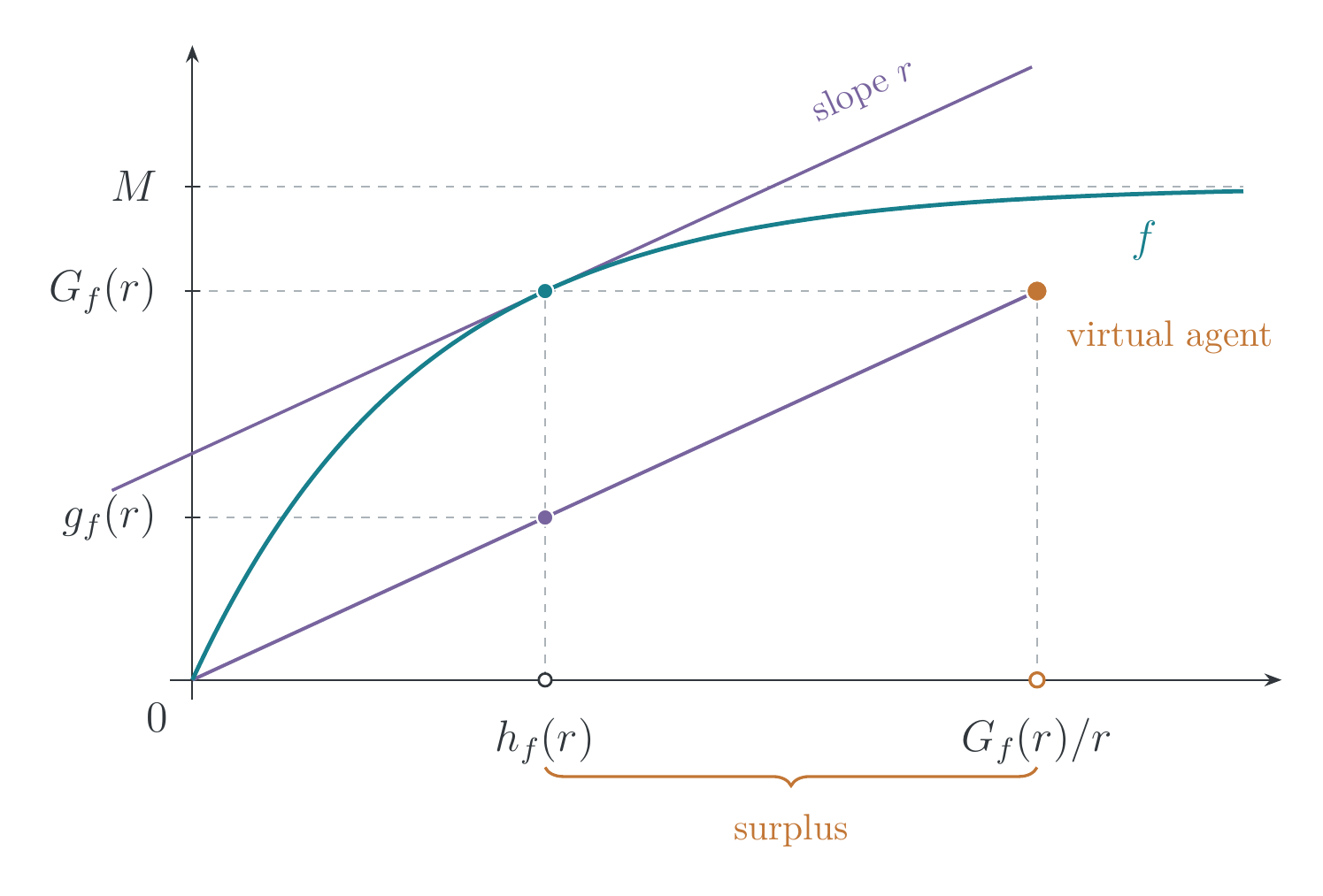}
  \caption{Selected virtual agent and its associated surplus}
  \label{figure:2}
\end{figure}

Next we introduce the function $F^{\cc C}$, which encodes the aggregate effect of trading  specified incoming amounts for each $c \in C$.

\begin{Definition}
   Given a system of external AMMs $\cc C = (C, \inn ,  \out , (f_c)_{c \in C})$ on $n$ tokens,
   define $F^{\cc C}: \R_{\ge 0}^C \longrightarrow \R^n$ by
   $$F^{\cc C}_i(x) = -\sum_{\inn (c) = i} x_c + \sum_{ \out (c) = i} f_c(x_c)  $$
\end{Definition}

\begin{Proposition} \label{prop:maximum}
  Given a system of external AMMs $\cc C = (C, \inn ,  \out , (f_c)_{c \in C})$ on $n$ tokens, and a price vector $p \in \PP_n$, the maximum of $p \cdot F^{\cc C}$ is attained at $h^{\cc C, p}$.
\end{Proposition}

\begin{proof}
\[
p \cdot F^{\cc C}(x)
= \sum_{c \in C} \left(p_{ \out (c)}f_c(x_c)-p_{\inn (c)}x_c\right)
= \sum_{c \in C} p_{ \out (c)}\left(f_c(x_c)-r_c x_c\right),
\qquad r_c=\frac{p_{\inn (c)}}{p_{ \out (c)}}>0
\]
It is sufficient to show that $x = h^{\cc C,p}$ maximizes every term $q_c(t) = f_c(t)-r_c t$.
If $r_c \ge f_c'(0)$, then $q_c'(t)=f_c'(t)-r_c<0$ for every
$t>0$, since $f_c'$ is strictly decreasing. Thus $q_c$ is maximized
at $0=h_c(r_c)$. If $0<r_c<f_c'(0)$, then $t=h_c(r_c)>0$ satisfies
$q_c'(t)=0$. The strict decrease of $f_c'$ implies
$q_c'(t)>0$ for $0 \le t<h_c(r_c)$ and $q_c'(t)<0$ for $t>h_c(r_c)$.
Hence $q_c$ is always maximized at $h_c(r_c)$.
\end{proof}

\section{Batch exchange with external AMMs} \label{section:batch-exchange}

In this section we consider a system of external AMMs together with a finite set of trade orders.
The fundamental theorem of Walraswap \ref{thm:main} establishes that, under a connectivity condition, there is a price vector at which all the trade orders can be executed together with external-AMM swaps so that the resulting batch is feasible and yields a nonnegative surplus of every token.
Moreover, the prescribed choice of external-AMM swaps maximizes the surplus
when measured according to the batch price vector (Remark \ref{remark:max-surplus}).

\begin{Definition} \label{def:trade-order}
 An admissible trade order is a pairwise supply function $\PP_n \xto \varphi \R^n$ that satisfies WGS.
\end{Definition}

By Propositions \ref{prop:characterization} and  \ref{prop:pairwise-WGS}, admissible trade orders are characterized by Definition \ref{def:WGS-function} and Equation \ref{eq:def-varphi}. Recall that continuous limit sell orders are {\it admissible trade orders}. A piecewise linear supply function, as defined in Examples \ref{examples}(III), is an admissible trade order if and only if all of the pieces $g_{a,b}$ satisfy $a, b \ge 0$.

Every nonzero admissible trade order connects its two tokens in at least one direction.
More precisely, if the order sells token $j$ and buys token $i$ at some price vector, then it connects $i \to j$. Indeed, WGS implies that the associated function $g$ and its conjugate are decreasing, so a positive value of either function implies a strictly positive limit at zero.

\begin{Definition}
Given a finite set $\cc O$ of admissible trade orders on $n$ tokens and a system $\cc C = (C, \inn ,  \out , (f_c)_{c \in C})$ of external AMMs on $n$ tokens, we say that $(\cc O, \cc C)$ connects indices whenever the directed graph on $n$ vertices determined by the edges
$$(i, j) \mbox{ such that there is } \varphi \in \cc O \text{ connecting } i \to j$$
$$(i, j) \mbox{ such that there is } c \in C \text{ with } \inn (c) = i,  \out (c)=j $$
is strongly connected.
\end{Definition}

\medskip
\noindent {\bf Notation:} For a finite set $\cc O$ of admissible trade orders on $n$ tokens, we consider $$\varphi^{\cc O} = \sum_{\varphi \in \cc O} \varphi$$

\begin{Remark} \label{remark:connects-strict}
 If $(\cc O, \cc C)$ connects indices, $\varphi^{\cc O} + \varphi^{\cc C}$ is strict, by Proposition \ref{prop:external-AMM} (5) and Proposition \ref{prop:sum-connects}.
\end{Remark}

\begin{Theorem}[Fundamental theorem of Walraswap] \label{thm:main}
Let $\cc O$ be a finite set of admissible trade orders on $n$ tokens and $\cc C = (C, \inn ,  \out , (f_c)_{c \in C})$ a system of external AMMs on $n$ tokens. If $(\cc O, \cc C)$ connects indices, there is a $p \in \PP_n$ such that

(a) For $i \in I_n$,
$$\left( \varphi^{\cc O}(p) + F^{\cc C}(h^{\cc C, p}) \right)_i =
\sum_{\substack{c\in C\\\inn (c)=i}} s_c \left( \frac{p_{\inn (c)}}{p_{ \out (c)}} \right) \ge 0
$$
where each term of the sum is nonnegative.

(b) Each term is positive if and only if $\frac{p_{\inn (c)}}{p_{ \out (c)}} < f_c'(0)$.
\end{Theorem}

\begin{proof} $ $
Consider the aggregate supply function
\[
 \Phi=\varphi^{\mathcal O}+\varphi^{\mathcal C}
\]
By Remark \ref{remark:connects-strict}, $\Phi$ is strict. The existence
theorem, Theorem \ref{thm:existence}, gives $p\in\mathcal P_n$ with
\[
 \Phi(p)=0
\]
Fix this $p$ and abbreviate
\[
 r_c=\frac{p_{\inn (c)}}{p_{ \out (c)}} \]
Definition \ref{def:va-and-not} gives
\[
 \varphi^{\mathcal C}_i(p)
 =-\sum_{\substack{c\in C \\ \inn (c)=i}}\frac{G_c(r_c)}{r_c}
  +\sum_{\substack{c\in C\\ \out (c)=i}}G_c(r_c)
\]
whereas the actual external-AMM trades give
\[
 F^{\mathcal C}_i(h^{\mathcal C,p})
 =-\sum_{\substack{c\in C\\ \inn (c)=i}} h_c(r_c)
  +\sum_{\substack{c\in C\\ \out (c)=i}} G_c(r_c)
\]
Consequently,
\[
\begin{aligned}
 \varphi_i^{\mathcal O}(p)+F_i^{\mathcal C}(h^{\mathcal C,p})
 &=\Big( \varphi^{\mathcal O}_i(p) + \varphi^{\mathcal C}_i(p) \Big) +
   \Big( F_i^{\mathcal C}(h^{\mathcal C,p}) -\varphi_i^{\mathcal C}(p) \Big) \\[4pt]
 &=\Phi_i(p)+\sum_{\substack{c\in C\\\inn(c)=i}}
       \left(\frac{G_c(r_c)}{r_c}-h_c(r_c)\right)\\[4pt]
 &=\sum_{\substack{c\in C\\\inn(c)=i}}s_c(r_c)
\end{aligned}
\]
Each term is nonnegative by virtue of Proposition \ref{prop:external-AMM} part (2).

(b) This follows from the strict inequality criterion in Proposition  \ref{prop:external-AMM} (2).
\end{proof}

\begin{Remark} \label{remark:max-surplus}
\normalfont
  In the context of the theorem, for any given price vector $p \in \PP_n$, the choice of external swaps $h^{\cc C, p}$ is the one that maximizes
  $$p \cdot \left( \varphi^{\cc O}(p) + F^{\cc C}(-) \right)$$
  by Proposition \ref{prop:maximum}. If the effective clearing price vector $p$ of a batch describes a market equilibrium in an exchange mechanism, the choice
  $h^{\cc C, p}$ implies surplus value maximization.
\end{Remark}

\begin{Remark} \label{remark:closely-related}
\normalfont
Our virtual agents are closely related to the strict-surplus trading rules of \cite[Appendix C]{Augmenting}. For a fee-free two-asset CFMM, the virtual allocation selects Rule E or F according to trade direction, and replacement by the actual external swap produces
the surplus described in \cite[Theorem C.6]{Augmenting}. Theorem \ref{thm:main} formulates the construction for external AMMs including their swap fees. They execute ordinary swaps, while the surplus remains in the batch, without requiring the pools to participate as uniform-price agents or subsequently withdraw surplus from their reserves.
\end{Remark}

\section{Computation of equilibrium} \label{section:convex-program}

For a supply function $\Phi$ that is the aggregate of admissible trade orders,
the equilibrium equation $\Phi(p) = 0$ is formulated here as a convex optimization problem in prices $p_i$ and signed traded values $y_{ij}$.\footnote{In particular, the construction applies to the aggregate of the trade orders and the virtual agents associated with external AMMs used in the proof of Theorem \ref{thm:main}.}
The traded value $y_{ij}$ may be interpreted as flux from token $i$ to token $j$. This makes sense for pairwise agents. The conservation of flux is a linear condition on $y$ that encodes equilibrium, while minimization of the objective function ensures that $y$ corresponds to $\Phi(p)$. The main idea of this construction is due to Devanur, Garg and V{\'e}gh \cite{Devanur2016} and was subsequently developed in \cite[\S 5]{Augmenting}.

By grouping together trade orders supported on the same pair of indices, we can write
\begin{equation} \label{eq:Phi}
\Phi = \sum_{(i,j) \in E_n} \varphi^{ij}
\end{equation}

\noindent where $E_n = \{(i, j): 1 \le i < j \le n \}$, and each $\varphi^{ij}$ is a pairwise supply function supported on $\{i, j\}$. We have
\begin{equation} \label{eq:varphi-ij}
\varphi^{ij}_k(p) = \begin{cases}
                       g^{ij}(\frac{p_j}{p_i}) & k = i \\
                       -\frac{p_i}{p_j} g^{ij}(\frac{p_j}{p_i}) & k = j \\
                        0 & k \neq i, j
                      \end{cases}
\end{equation}
for a system of WGS functions $\{g^{ij}\}_{(i, j) \in E_n}$ (recall Definition \ref{def:WGS-function}).

\medskip
To describe the convex program, we need the following lemma.

\begin{Lemma} \label{lemma:aux-convex-functions}
  Given a WGS function $\R_{>0} \xto g \R$,
define two functions associated to $g$:
\begin{align*}
\R_{>0} \xtoo H \R, \ \ \ \ \ & \ \ \ \ H(r) = -\int_1^r \frac{g(s)}{s} ds \\
\R \xtoo Z [0, +\infty], & \ \ \ \ Z(t) = \sup_{u > 0} \{ - t \log u - H(u) \}
\end{align*}

(1) $H$ and $Z$ are convex functions.

(2) $Z(t) = - t \log r - H(r)$ if and only if $t = g(r)$.

(3) The function $\PP_2 \times \R \xtoo F \R \cup \{+\infty\}$ defined by $F(p, q, y) = p H(\frac{q}{p}) + p Z(\frac{y}{p})$ is convex.

\end{Lemma}

\begin{proof}
Note that $H(1) = 0$. By choosing $u = 1$, we get $Z(t) \ge 0$, so $Z$ is
well-defined.

(1) $H$ is convex because its derivative $H'(r)=-g(r)/r$ is increasing.
$Z$ is convex because it is defined as a supremum of affine functions.

(2) Defining $f_t(u) = -t \log u - H(u)$ for $u>0$, we have
$$f_t'(u) = \frac{1}{u}(g(u) - t)$$
If $r$ is a maximizer of $f_t$, then $f_t'(r)=0$, so $g(r) = t$.

\noindent Conversely, if $g(r)=t$, since $g$ is decreasing, $f_t'$ is at least $0$ for $u < r$ and at most $0$ for $u > r$. Therefore, $f_t$ attains the maximum at $u=r$.

(3) Each term is convex because it is the {\it perspective} of a convex function.
\end{proof}

\begin{Lemma} \label{lemma:By-zero}
 Let $\PP_n \xto \Phi \R^n$ be a supply function that is the aggregate of admissible trade orders, described by the system of WGS functions $\{g^{ij}\}_{(i,j) \in E_n}$.
 Let $B \in \R^{I_n \times E_n}$ be the matrix
 $$\hspace{3cm} B^{ij}_k = \delta_{ik} - \delta_{jk} \hspace{.5cm} (k \in I_n, \ ij \in E_n)$$

 \vspace{0.1cm}
 \noindent For every $p \in \PP_n$, consider the values $y_{ij} = p_i g^{ij}(\frac{p_j}{p_i})$.
 Interpreting $y$ as a column, $y \in \R^{E_n \times 1}$,
 $$\Phi(p) = 0 \Longleftrightarrow By = 0$$
\end{Lemma}

\begin{proof}
 $$\Phi_k(p) = \sum_{ij \in E_n} \varphi^{ij}_k(p) =
 \sum_{k < j \le n} g^{kj} \left(\frac{p_j}{p_k} \right) -
 \sum_{1 \le i < k} \frac{p_i}{p_k} g^{ik} \left( \frac{p_k}{p_i} \right) =
 \frac{1}{p_k} \left( \sum_{k < j \le n} y_{kj} -  \sum_{1 \le i < k} y_{ik} \right) $$
 $$
 (By)_k = \sum_{1 \le i < j \le n} B^{ij}_k y_{ij} =
 \sum_{1 \le i < j \le n} (\delta_{ik} - \delta_{jk}) y_{ij} =
 \sum_{k < j \le n} y_{kj} - \sum_{1 \le i < k} y_{ik} = p_k \Phi_k(p)
 $$
\end{proof}

\noindent {\bf Notation:} $S_n = \{ y \in \R^{E_n} : By = 0 \}$ for $B$ as in the lemma.

\begin{Theorem}\label{thm:convex-program}
Let $\Phi$ be a supply function that is the aggregate of admissible trade orders.
We follow the notation of equations (\ref{eq:Phi}) and (\ref{eq:varphi-ij}).
Consider $\PP_n \times S_n \xtoo F \R \cup \{+\infty\}$ defined by
  $$F(p,y) = \sum_{ij \in E_n}
  \left[ p_i H_{ij} \left( \frac{p_j}{p_i} \right) +
  p_i Z_{ij} \left(\frac{y_{ij}}{p_i} \right) \right]$$
where $H_{ij}$ and $Z_{ij}$ are the functions associated to $g^{ij}$ according to Lemma \ref{lemma:aux-convex-functions}.

(0) The function $F$ is jointly convex in $(p,y)$.

(1) $F(p, y) \ge 0$.

(2) $F(p, y) = 0$ if and only if $\Phi(p) = 0$ and $y_{ij} = p_i g^{ij}(\frac{p_j}{p_i})$.
\end{Theorem}

\begin{proof} $ $

(0) Joint convexity follows by summing the perspectives in
Lemma \ref{lemma:aux-convex-functions} (3) and restricting to the linear constraint $By=0$.

(1) For $r \in \R_{>0}, \ t \in \R$, consider
$$D_{ij}(r,t):= t \log r +  H_{ij}(r) + Z_{ij}(t) \ge 0$$
The inequality follows from the definition of $Z_{ij}$.
Now we write $F$ in terms of $D_{ij}$
$$F(p, y) = \sum_{ij \in E_n} p_i D_{ij} \left( \frac{p_j}{p_i}, \frac{y_{ij}}{p_i} \right)
- \sum_{ij \in E_n} y_{ij} (\log p_j - \log p_i)$$
The negative sum equals the matrix product $(\log p) B y = 0$.
Here $\log p$ is a row and $y \in S_n$ a column. Therefore,
$$F(p, y) = \sum_{ij \in E_n} p_i D_{ij} \left( \frac{p_j}{p_i}, \frac{y_{ij}}{p_i} \right) \ge 0$$

(2) If $F(p, y) = 0$, by Lemma \ref{lemma:aux-convex-functions} (2) and the last expression, we have $\frac{y_{ij}}{p_i} = g^{ij} \left( \frac{p_j}{p_i} \right) \ \forall ij \in E_n$.
Now Lemma \ref{lemma:By-zero} gives $\Phi(p)=0$.
Conversely, if $\Phi(p)=0$ and $y_{ij} = p_i g^{ij} \left( \frac{p_j}{p_i} \right) \ \forall ij \in E_n$, Lemma \ref{lemma:By-zero} gives $By=0$, and Lemma \ref{lemma:aux-convex-functions} (2), together with the expression for $F$ in terms of $D_{ij}$, gives $F(p, y) = 0$.
\end{proof}

Finding an equilibrium vector $p \in \PP_n$ is therefore equivalent to finding a
point of objective value zero in the convex optimization problem
\[ \begin{aligned}
 \underset{p \in \stackrel{\circ}{\Delta}_n, y \in S_n}{\operatorname{minimize}}\quad
 &\sum_{ij\in E_n}\left[
 p_i H_{ij}\!\left(\frac{p_j}{p_i}\right)
 +p_i Z_{ij}\!\left(\frac{y_{ij}}{p_i}\right)\right]
\end{aligned}
\]
where $\stackrel{\circ}{\Delta}_n$ is the interior of the $(n-1)$-simplex.

\begin{Remark} \normalfont
The objective function $F$ and its subgradients can be computed efficiently for both piecewise linear trade orders and the virtual agents corresponding to Uniswap's concentrated liquidity pools. The input orders can be preprocessed into
aggregate piecewise representations. Subsequent objective and subgradient evaluations require only logarithmic-time interval lookup, without rescanning the individual orders.
The details of this, as well as the design and performance of a practical equilibrium solver, are left to future work. See also \cite[Appendix F]{Speedex}.
\end{Remark}

\section{Application to an exchange protocol} \label{section:mechanism}

\subsection{Introduction}

The relevance of Theorem \ref{thm:main} is tied to its applicability to an exchange system that could eventually be adopted in practice.
This section outlines Walraswap, a proposed decentralized exchange protocol intended to put the theorem's construction into practice.
As explained in the introduction, the protocol is designed to operate on the Ethereum blockchain.

A system that makes use of the scheme of Theorem \ref{thm:main} must be based on batches with uniform prices for internal orders. This is consistent with \cite{Speedex} and \cite{Augmenting}. However, we need to emphasize a significant departure from their framework. {\it Constant Function Market Makers} (CFMMs) in \cite{Augmenting} are {\it internal AMMs}. This means that they are agents participating in the exchange activity of the protocol, and the uniform prices also apply to them. However, their theoretical treatment differs from that of regular traders. CFMMs' trade response is treated as a protocol feature, and the impossibility results \cite[Theorem 1.5, Theorem 1.6, Proposition 3.12]{Augmenting} are viewed as design trade-offs. In Walraswap, internal AMMs participate through admissible trade orders in the sense of Definition \ref{def:trade-order}, on the same footing as ordinary traders. In contrast with \cite{Augmenting}, their trade response is not a protocol feature but the decision of the agents providing the liquidity. This point of view guides Walraswap's protocol design.

When describing the protocol, we use the present tense to express its objective features. Very often we also need to describe the expected emergent behavior.

\subsubsection{Basic batch structure} \label{subsubsec:basic-batch-structure}

The protocol executes batches of exchange operations, at most one batch per block.
There are two types of exchange operations:

\medskip

(a) Trade orders submitted to the system itself.

(b) Swaps at external AMMs.

\medskip

Type (a) encompasses messages generated and signed by users, as well as orders tacitly generated by smart contracts specially designed for this purpose.
We call such contracts {\it internal AMMs} or {\it affiliated AMMs}. Both subtypes are modeled by {\it pairwise WGS supply functions} (Definition \ref{def:trade-order}), so the trade orders satisfy the continuity and monotonicity conditions. Most of the
user-generated trade orders are expected to be continuous limit sell orders (Remark \ref{remark:approximate-limit-order}).

All trade orders of type (a) are settled at prices determined by a single vector $p \in \PP_n$ associated with the batch.
They have to pay the protocol trading fee (\ref{subsec:trading-fees}), and the gas fee (\ref{subsec:gas-fees}). The price vector $p$ does not apply to external swaps.

The prescription of Theorem \ref{thm:main} will allow the protocol to take advantage of external AMMs, incorporating their liquidity into the protocol
without granting them the protection provided by the batch's price vector.
As a matter of fact, every external AMM included in a batch, executing a nonzero swap, will receive a strictly worse price.

Apart from differences in order format, CoW Protocol shares this basic batch structure. It is therefore useful to explain its main difference from Walraswap, which explains why the prescription of Theorem \ref{thm:main} does not apply to CoW. Historically, CoW Protocol aimed to choose a batch maximizing the {\it surplus} for trade orders. The surplus is defined as the difference between the limit amount set by the user and the effective execution amount.
This objective does not require the selected batch to implement a Walrasian equilibrium for the full set of submitted orders; the batch can omit orders whose prescribed responses would be nonzero at equilibrium. Because surplus is measured relative to the declared limit, this objective can favor orders with more permissive limits. The current version of CoW Protocol, in addition, no longer requires a single clearing-price vector for all trades in a batch auction.\cite{CoWCIP67}

\subsubsection{Main difficulty}

A batch mechanism with the structure described above does not automatically achieve the goal of eliminating or substantially reducing MEV in decentralized exchange activity.
The reason is that
any entity responsible for submitting batches will always
have the ability to censor some of the orders to manipulate the outcome and potentially obtain profit. Nevertheless, with the right design, such censorship can be strongly discouraged in a repeated-game setting.

Throughout the text, we call an entity that forms, signs and submits a blockchain transaction to settle a batch a {\it solver}.

The informal note \cite{Affiliated} discusses {\it permissionless solving},
in which anyone can be a solver, and asks whether this naive choice may create incentives for maximal inclusion through the expected use of
private channels. In this setting, however, it looks difficult to audit the behavior of the solvers while preserving privacy, and to coordinate the migration away from an established dishonest solver. Moreover, another complication arises. The affiliated AMMs remain public, and consequently they do not have an interaction channel with the solvers. This could be overcome by imposing, at the smart contract level, mandatory inclusion of all the affiliated AMMs in the batches. However, this is impractical when we consider gas fees: if affiliation is unregulated, the gas cost may rise unsustainably.
At this point, we may well stop exploring this winding road and start afresh on a better path.

The following subsections propose a solution to the main difficulty. The presentation is intended to establish the plausibility of such a protocol. To this end, we aim to strike a good balance between concreteness and flexibility. Whenever we choose to specify design details, they should be understood as illustrative rather than definitive.

\subsection{Governance}

A fixed set of solvers is incompatible with decentralization. For this reason, the protocol has a mechanism to govern the solvers.

The protocol is governed by a decentralized autonomous organization (DAO) composed of holders of its governance token, called the $W$ token. Voting rights are proportional to token holdings. The DAO controls the solver allowlist and determines protocol parameters.
Token holders have an economic interest in protocol revenue through the trading fee mechanism described below.

The smart contract that enables the exchange mechanism, implementing the basic batch structure \ref{subsubsec:basic-batch-structure}, is called
{\it $W$ engine}.
The $W$ engine is designed to be upgradeable, with upgrades authorized by the DAO, to improve the protocol.
This feature ensures that the success of the protocol is not tied to a particular implementation, and allows the community to self-organize around a cause rather than a fixed mechanism.

\subsection{Allowlist of solvers}

Recall that a {\it solver} is an entity that forms, signs and submits a blockchain transaction to settle a batch.
The Walraswap protocol uses an allowlist of {\it permissioned solvers}. To support decentralized governance, the allowlist is stored in the $W$ engine and
it is controlled by the direct vote of the $W$ token holders.
The precise details of the voting mechanism are dispensable here. However, we may assume that a solver needs a minimum threshold of upvotes, and more upvotes than downvotes to appear on the allowlist. Of course, votes are always token-weighted.

The solvers receive a payment for every successful batch. This payment should be the minimum required for solvers to prefer adhering to their commitments over violating them for a larger short-term gain.
The exact solver payment amount is not encoded in the $W$ engine. However, an upper bound is actually enforced by the $W$ engine (see \ref{solver-payout} below).

\subsection{Definition of solver honesty} \label{subsec:honesty}

Recall that the goal of the protocol is to systematically execute all the feasible existing orders up to a certain time, at an equilibrium price vector. What we show in this subsection is that there exists a notion of {\it honest behavior} for the solvers that is compatible with this goal and can be assessed by the token holders. Solvers on the allowlist will need to behave honestly in order to preserve their position. {\it Honest behavior} is enforced through governance, and its precise definition may be subject to discussion and evolution.
We emphasize that the definition below is intended as one plausible formulation, based on the idea that full disclosure of trade orders enables the community to assess solvers’ behavior and hold them accountable.

We say that a solver {\it behaves honestly during the interval between two consecutive blocks} if

\medskip
$(1)$ It discloses every valid order it sees in real time up to a certain cutoff time $t_0$. More precisely, it publishes a signed cryptographic chain of messages, each containing the latest set of valid orders, at a frequency $\nu$ of the order of $1$ per $100 ms$ ($\nu \sim 10 s^{-1}$).

$(2)$ It applies a predetermined public deterministic algorithm $A$ to the full set of valid orders to produce a valid batch. The input also includes the states of the required external AMMs.

$(3)$ It emits a signed transaction corresponding to that batch before $t_1$, or nothing if the batch is empty.

\medskip
The precise definition of {\it valid order} is solver dependent. Denoting the validity policy by $V$, we specify a solver strategy by $(V, \nu, t_0, t_1, A)$.

We say that a solver {\it behaves honestly} if it publicly commits to a strategy sufficiently far in advance and follows it.

Let $\cc O$ denote the final set of valid orders disclosed by the solver up to the cutoff $t_0$ for a given block slot.

The algorithm $A$ includes:
\medskip

$\sbt$ The determination of a subset $\cc O' \subset \cc O$ that is the set of trade orders effectively executed in the resulting batch.
The need to select a subset is explained in the discussion of gas fees in Subsection \ref{subsec:gas-fees} below.

$\sbt$ The precise formula to determine the solver payout. See Subsection \ref{solver-payout} for further details.

The DAO will assess whether that strategy promotes broad order inclusion and equilibrium clearing. In \ref{subsec:why-batches} we explain why we expect the algorithm $A$ of an authorized solver to apply the construction of Theorem \ref{thm:main} to the set $\cc O'$.

Beyond the solver's declared strategy and adherence to it, the DAO may also consider technical and operational factors, such as geographic location, network connectivity, and the quality and reliability of its infrastructure. These factors may affect solver performance and the resilience and decentralization of the protocol.

\subsection{Censorship resistance}

Anyone can verify the execution of algorithm $A$ once its input is known. An incorrect application will be evident, since the input will be public and signed by the solver.

If a valid order is omitted, participants in the community-operated peer-to-peer network may signal it in real time using more or less sophisticated methods.

Whenever a solver fails to honor its commitment, token holders will vote to remove it from the allowlist.
Because of latency and network imperfections, assessments of ``solver honesty'' will be slightly subjective. Nevertheless, we expect solvers to be publicly accountable and therefore behave honestly in order to keep their business going. A solver's identity will have to be public. Otherwise, it could misbehave and reappear at no cost. We expect persistent, economically significant violations of a solver's commitments to become detectable over time.

\subsection{Allowed tokens} \label{subsec:allowed-tokens}

In a fixed-point search, the number of tokens is the dimensionality of the problem. The convergence time might be impractical for a sufficiently large number of tokens. For this reason, Ramseyer, Goel, and Mazières anticipate a core-periphery trading topology: a set of interconnected numeraire tokens, together with peripheral tokens traded only against one numeraire \cite[{\it Linear Program Scalability} and Appendix E]{Speedex}.
This topology allows the system to support an {\it a priori} unlimited number of peripheral tokens,
because the search for equilibrium prices can be parallelized using Proposition \ref{prop:decomposition} or, similarly, \cite[Theorem 5]{Speedex}.

More generally, Proposition \ref{prop:decomposition} allows disjoint clusters of peripheral tokens, each connected to the rest of the trading network only through a single core token. A peripheral token traded against several core tokens could be represented by separate copies, one for each core-token market, with orders assigned to the corresponding copy. Each copy would clear at a potentially different price.

The set of {\it core tokens} is governed by the DAO, while the list of allowed pairs is decided by the solver's validity policy $V$. The $W$ token belongs to the set of {\it core tokens}.

In this setting, the connection-of-indices condition in Theorem \ref{thm:main} is mild. It is inexpensive for the DAO to guarantee connectedness for the set of core tokens: it suffices to set up a spanning tree of affiliated AMMs based, for instance, on the product-maximizing rule from Example \ref{ex:CP-AMM}.
For peripheral tokens, the solver can simply prune the ones that are not strongly connected.

\subsection{Trading fees (protocol fees)} \label{subsec:trading-fees}

By default, each executed trade order pays a trading fee according to the executed volume. The simplest design choice is to make the fee value proportional to the volume.
\begin{equation} \label{eq:fee-value}
\mbox{trading fee value } = \ \tau \cdot \mbox{ value of sold tokens}
\end{equation}
The supply function specifies the exchange excluding fees. The fee rate $\tau$ is determined by the $W$ engine, and directly governed by the $W$ token holders, who are also the beneficiaries.

The trading fee must be paid in a {\it core token} and is collected by the $W$ engine. Periodically, the $W$ engine swaps the collected fees for $W$ tokens and burns its $W$ reserves, thereby transferring the proceeds to the token holders. We leave the precise mechanism for this unspecified.

\medskip
{\bf Subsidies} \label{trader-incentives}

The DAO may want to subsidize specific types of trade orders or all of them for a period of time. This could be done by reducing or waiving trading fees, or even by issuing $W$ tokens to eligible traders.

A particularly relevant case is that of {\it affiliated AMMs}. Lower fees could encourage liquidity providers to participate directly in Walraswap, increasing the liquidity available at common batch prices and reducing dependence on external swaps.

\subsection{Gas fees} \label{subsec:gas-fees}

{\bf Description of the difficulty}
\medskip

Executing an order entails a gas cost that does not vanish with the
traded amount. Charging this cost in full could make a small partial
execution disproportionately expensive for the trader. Consequently,
economic viability may require excluding some orders from the full set of valid orders to form the batch, as anticipated in Subsection \ref{subsec:honesty}.

\medskip
{\bf How the protocol deals with gas fees}
\medskip

Each executed trade order pays a gas fee. The $W$ engine assigns a standard gas-unit charge to each order type.

To protect traders, each trade order specifies an upper bound for the dimensionless quantity
$$\frac{\text{gas fee}}{\text{value of sold tokens}}$$
We call this bound the maximum relative gas fee.
The trade order also specifies the token used to pay this fee. This choice can be accepted or rejected under the solver's validity policy $V$.
Unit conversion is performed using the batch's clearing price vector.

The solver charges the gas fee through the $W$ engine as the minimum of {\it (a) the tabulated gas units times the transaction's effective gas price} and {\it (b) the value derived from the maximum relative gas fee}.
Small shortfalls in gas fees relative to execution costs on some orders may be offset by small surpluses on others.

The solver's algorithm $A$ chooses a subset $\cc O'$ of the set of valid orders $\cc O$ and forms a batch. Once $\cc O'$ is chosen, the expected resulting batch has the form described by Theorem \ref{thm:main} (see Subsection \ref{subsec:why-batches}). A good solver algorithm should exclude an order from the set of valid orders only when it cannot be made viable, or when its gas usage may exceed the tabulated estimate. The latter category might be partially or totally filtered out at an earlier stage by the validity policy $V$.

\subsection{Surplus destination} \label{surplus-destination}

By Theorem \ref{thm:main}, we expect to observe surplus in settled batches, generated by swaps at external AMMs. There are in principle four possible recipients for the surplus: the DAO, the block builder, the solver, and traders.

We exclude traders from direct surplus distribution to keep the transfers simple and settlement efficient. In a context where the protocol receives surplus consistently, a part of it could be indirectly transferred to the traders through trading fee discounts or negative trading fees, as discussed in Subsection \ref{trader-incentives}.

The surplus is first received by the $W$ engine. Then, a part of it is transferred to the solver through the solver payout (see \ref{solver-payout} below), and the solver will  fine-tune a payment to the block builder, according to the market conditions and following publicly committed rules.

We expect the DAO to capture some of the value associated with external AMMs' stale prices through the surplus. However, another fraction needs to be offered by the solver to the block builder, since the batch will be competing against other transactions for inclusion.

\subsection{Solver payout} \label{solver-payout}

When settling all trades, the $W$ engine measures the collected trading fees and swap surpluses separately.
The solver is entitled to fractions $\alpha_1$ and $\alpha_2$ of these two revenue streams, respectively, calculated separately for each token.
The numbers $\alpha_1, \alpha_2$ are specified by the solver's transaction, but should agree with the solver's strategy. The $W$ engine enforces the conditions $\alpha_i \le \gamma_i$ ($i = 1, 2$) for protocol-governed parameters $0 < \gamma_1 < 1$, $0 < \gamma_2 \le 1$.

The DAO should keep most of the trading fees, so we expect $\gamma_1$ to be small.
The cap on the surplus fraction may reach one, allowing the solver to offer a competitive payment to the block builder.

\subsection{Why batches are expected to follow the Fundamental Theorem \ref{thm:main}} \label{subsec:why-batches}

Consider an agent allowed to perform a swap at each AMM in a system of external AMMs $\cc C$. By Proposition \ref{prop:maximum}, if the agent has a linear utility function defined by a price vector $q$, the optimal swap amounts are given by $h^{\cc C, q}$. This fact is actually more general. Assume that the agent has a differentiable utility function and positive marginal token utilities. The optimal swap amounts are $h^{\cc C, q}$, where $q$ is the vector of marginal token utilities evaluated at the agent's final holdings. This marginal-utility characterization of optimal swaps also appears in \cite[\S3]{CFMMRouting}.

This suggests that the batches of a solver constrained only by the $W$ engine's uniform-price condition should be described by $(p, q)$,
where $p$ is the clearing price vector and $q$ determines $h^{\cc C, q}$ as above. Assuming that the solver seeks to maximize the expected market value of the proceeds, the vector $q$ should approximate the expected market prices immediately after execution.

For this reason, the condition $p=q$ used in Theorem \ref{thm:main} would be a natural way to couple external markets with Walraswap's clearing price vector. It is an additional fairness condition that we expect the DAO to impose on solvers.

\subsection{Discussion on decentralization}

The degree of decentralization of this protocol depends in part on the distribution of the $W$ token. Initially, we expect governance power to be highly concentrated and to become more decentralized over time.

The solver role is expected to become concentrated in a small number of specialized entities. Solvers compete for a limited profit, and advantages from specialization, scale, and access to superior infrastructure may favor a small group of efficient operators. This dynamic is analogous to the concentration observed among block builders.

The $W$ token is the tool that will allow the community to select solvers that behave fairly.
This market-regulated power structure can therefore support decentralized control of the exchange.

\subsection{Competition between batches}

Different active solver strategies can produce different valid batches for each block slot, but at most one batch per block is admitted by the $W$ engine.
The block builder selects one batch transaction for inclusion or includes none.

According to the previous discussion, we expect a small number of specialized entities to control most of the active solvers. Each of these entities might control several solvers simultaneously.

In a scenario where Walraswap enjoys substantial adoption, its batches could represent a significant share of blockspace and block reward.
This could create incentives for the entities controlling the solvers to integrate their operations with block builders, allowing closer coordination between batch formation and block construction.
In particular, the same entities might operate both solvers and block builders.
Initial adoption of Walraswap can be encouraged by an incentive program.

\subsection{Economic properties of Walraswap}

The Walraswap protocol rests on two pillars:
\medskip

(a) Achieving maximal inclusion subject to practical constraints.

(b) Implementing a Walrasian equilibrium.

\medskip
The practical constraints arise from network imperfections
and from gas costs, which may require selecting a subset
of the valid orders $\cc O' \subset \cc O$, as explained in Subsection \ref{subsec:gas-fees}.

Following \ref{subsec:why-batches} and Theorem \ref{thm:main}, a batch generated by an honest solver selected by the DAO is expected to implement a Walrasian equilibrium among $\cc O'$ and the virtual agents corresponding to some system of external AMMs $\cc C$. Given $\cc O'$ and $\cc C$, the equilibrium is unique, in the sense of executed swap amounts, by virtue of Theorem \ref{thm:net-trades}.

An honest solver that also trades may enjoy a last-mover advantage when choosing its own orders. However, this hypothetical advantage does not undermine the two pillars.

Economically, these two properties make traders and virtual agents each other's direct counterparties, apart from gas and trading fees. Walraswap can be understood as a mechanism for strengthening the collective bargaining power of traders and liquidity providers \cite{Affiliated}. By bringing their orders and liquidity into a common exchange mechanism, participants can condition access to this activity on common execution rules. These rules are enforced either by the smart contract or by governance. With sufficient adoption, this coordination could yield equilibrium prices representative of the broader market, creating a positive feedback loop between users’ confidence and price discovery.

This market structure has the potential to be highly efficient and competitive. It could attract a substantial share of decentralized exchange volume on Ethereum, as well as activity currently taking place on centralized exchanges.

\appendix

\section{Non-differentiable external AMM functions} \label{app:non-differentiable}

This appendix generalizes the construction and results from Sections \ref{section:external-amms} and \ref{section:batch-exchange} to AMM functions that are not necessarily differentiable.
Condition (d) of Definition \ref{def:ammf} is replaced by a concavity condition. The paradigmatic case where this generalization is needed is that of Uniswap's concentrated liquidity pools containing regions with no liquidity. In this situation, the right derivative jumps from the end price of a region with available liquidity to the starting price of the next region with available liquidity or to zero if there is no such next region.

\begin{Definition} \label{ammf}
 An AMM function is a map $\R_{\ge 0} \xto f \R_{\ge 0}$ with the following properties:

 (a) It is increasing.

 (b) $f(0) = 0$, but $f \not\equiv 0$.

 (c) It is bounded above.

 (d) It is concave. Moreover, it is strictly concave on $f^{-1}\big([0, M)\big)$, where
 $M = \sup\operatorname{Im}(f)$. Namely, for $x, y \ge 0$, $x \neq y$, such that $f(x), f(y) < M$, we have
 $$f(tx + (1-t)y) > tf(x) + (1-t)f(y) \ \ \ \forall \ 0 < t < 1$$

 (e) $\partial_+ f(0) < \infty$.
\end{Definition}

In (d), we require strict concavity except in the region where $f$ attains its maximum. The strictness is necessary for the continuity of $h_f$.

The following proposition collects the properties needed to extend the definitions of $h_f$, $g_f$, and $G_f$ (Definitions \ref{def:h_f} and \ref{def:g-and-G}) and to generalize Proposition \ref{prop:external-AMM}. Once these properties are established, the proofs are straightforward.

\begin{Proposition} \label{prop:concavity}
 Let $\R_{\ge 0} \xto f \R_{\ge 0}$ be an AMM function,
 and $M = \sup\operatorname{Im}(f)$.
 The following properties hold:

 \medskip
 (I) The right derivative exists at every point $x \ge 0$. The left derivative exists at every point $x > 0$. The function $f$ is continuous.
 Moreover, $\R_{\ge 0} \xto{\partial_+ f} \R_{\ge 0}$ is decreasing and right continuous.

 \medskip
 (II) $\partial_+ f$ is strictly decreasing on $f^{-1}\big([0, M) \big)$.
 Moreover, for $x < y$ both in $f^{-1}\big([0, M) \big)$, we have
 $\partial_- f(y) < \partial_+ f(x)$.

 \hspace{0.47cm} $\lim_\infty \partial_+ f = 0$.

 \medskip
 (III) The map $\R_{> 0} \xto{h_f} \R_{\ge 0}, \ h_f(r) = \min \{ x \ge 0 / \partial_+ f(x) \le r \}$ is decreasing and continuous.

 \medskip
 (IV) $h_f(\partial_+f(x)) = x$ for every $x$ such that $f(x) < M$.

 \hspace{0.67cm} $h_f(r) = 0$, for every $r \ge \partial_+ f(0)$.

 \medskip
 (V) $r h_f(r) \le f(h_f(r))$ for every $r > 0$.

 \hspace{0.53cm} The inequality is strict whenever $h_f(r) > 0$, or equivalently for $r < \partial_+ f(0)$.

  \medskip
  (VI) $\lim_{r \to 0} f(h_f(r)) = M$.

   \medskip
  (VII) $\lim_{r \to 0} r h_f(r) = 0$.
\end{Proposition}

For the sake of brevity we do not include all the details in the proof.

\begin{proof}

The properties in (I) are standard properties of concave functions.
For the convenience of the reader, we outline a proof of these properties.

Whenever we have three numbers $0 \le x_1 < x_2 < x_3$, by the concavity condition,
 the point $B = (x_2, f(x_2))$ is not below the line through the points $A = (x_1, f(x_1))$ and $C = (x_3, f(x_3))$. Thus, we have
 $\operatorname{slope}(BC) \le \operatorname{slope}(AC) \le \operatorname{slope}(AB)$.
 The existence of the one-sided derivatives can be proved from this fact.
 From the same fact it follows that $\partial_+ f$ is decreasing.
 $\partial_+ f \ge 0$ because $f$ is increasing. Continuity of $f$ follows from the existence of one-sided derivatives.

Right continuity of $\partial_+ f$ can be proved with the help of a statement similar to Lagrange's mean value theorem: for $r_1 < r_2$ there is an $r$ in between such that $\partial_+ f(r)$ is greater than or equal to the slope of the line through $(r_1, f(r_1))$ and $(r_2, f(r_2))$.

\medskip
 (II) Having $\partial_+ f(x) = \partial_- f(y)$ for some $x < y$ with $f(y) < M$ would imply that $\partial_+ f$ is constant in $(x, y)$. In this case we would have $f'$ equal to a constant on that interval, contradicting strict concavity.
 $\lim_\infty \partial_+ f = 0$ is a consequence of the assumption that $f$ is bounded above.

 \medskip
 (III), (IV) For $r > 0$, write
 $$T_r = \{ x \ge 0 / \partial_+ f(x) \le r \}$$
Since $\partial_+ f$ is decreasing, $T_r$ is an interval.
 We have $T_r \neq \emptyset$ thanks to $\lim_\infty \partial_+ f = 0$. The infimum of $T_r$ is a minimum because $\partial_+ f$ is right continuous.
 We find that $h_f$ is decreasing because $T_{r_1} \subset T_{r_2}$ if $r_1 \le r_2$. Now we can check (IV). Take $x$ such that $f(x) < M$.
 Since $x \in T_{\partial_+ f(x)}$, we have $h_f(\partial_+ f(x)) \le x$. But $x$
 is also a lower bound for $T_{\partial_+ f(x)}$, because of the strict
 decrease of $\partial_+ f$. Thus, $h_f(\partial_+ f(x)) = x$.

 The other assertion in (IV) follows
 from $T_r = \R_{\ge 0}$ for $r \ge \partial_+ f(0)$.

 To show continuity of $h_f$ at $r > 0$, we consider three cases.
 First assume $0 < f(h_f(r)) < M$. Consider an interval $(x^-, x^+) \ni h_f(r)$ such that $f(x^+) < M$.
 We have $r \in (\partial_- f(x^+), \partial_+ f(x^-))$ and $h_f\big( (\partial_- f(x^+), \partial_+ f(x^-)) \big) \subset (x^-, x^+)$ by the first part of (IV).

 For the case $f(h_f(r)) = M$, if $0 < r' < r$, then $h_f(r')=h_f(r)$. Thus $h_f$ is left continuous. Right continuity can be handled as in the first case.
 Finally, consider $f(h_f(r)) = 0$. Strict concavity and monotonicity imply $h_f(r)=0$. The rest of the proof is similar to the argument in the first case.

 \medskip
 (V) Equality holds trivially for $h_f(r)=0$. Now assume $h:=h_f(r) > 0$.
 We have
$$r \le \lim_{x \uparrow h} \partial_+ f(x) = \partial_- f(h)$$
The line $l$ through the point $(h, f(h))$ with slope $\partial_- f(h)$ lies above the graph of $f$.
 $$l(x) = \partial_- f(h) (x - h) + f(h)$$
 Evaluating at $0$ we have $l(0) > 0$, or $f(h) > \partial_- f(h) h \ge r h$.

  \medskip
  (VI) Fix $0 < K < M$. Choose $x > 0$ such that $K < f(x) < M$.
  By (III) and (IV), for $0 < r \leq \partial_+ f(x)$,
\[ h_f(r)\geq h_f(\partial_+ f(x))=x \]
Since $f$ is increasing,
\[ K<f(x)\leq f(h_f(r))\leq M \]

  \medskip
  (VII) Fix $0 < K < M$.
  Choose $x_1 > 0$ such that $K < f(x_1)$. By the definition of $h_f$ and concavity, the line through $(h_f(r),f(h_f(r)))$ with slope $r$
  $$l(x) = r(x - h_f(r)) + f(h_f(r))$$
  lies above the graph of $f$. Evaluating it at $x_1$ gives
$$r h_f(r)\le f(h_f(r))-f(x_1)+r x_1 \le M-f(x_1)+r x_1 $$
Set $r_0 := \frac{f(x_1) - K}{2x_1} > 0$.
Then, for every $0 < r \le r_0$,
$$ r h_f(r) \le M - \frac{f(x_1) + K}{2} < M - K $$
Since $K < M$ is arbitrary, this proves $\lim_{r \to 0}r h_f(r)=0$.
\end{proof}

Finally, to generalize Proposition \ref{prop:maximum}, we need the following lemma.

\begin{Lemma}
 Let $f$ be an AMM function and $r > 0$. The maximum of $\R_{\ge 0} \xto q \R$, $q(x) = f(x) - rx$, is uniquely attained at $h_f(r)$.
\end{Lemma}

\begin{proof}
 The line through $(h_f(r),f(h_f(r)))$ with slope $r$
  $$l(x) = r(x - h_f(r)) + f(h_f(r))$$
  lies above the graph of $f$. Therefore,
  \[ f(x) \le r(x - h_f(r)) + f(h_f(r)) \]
  \[ f(x) - rx \le f(h_f(r)) - rh_f(r) \]
  showing that $h_f(r)$ is a maximizer.

  Uniqueness follows from the observation that $q$ is strictly concave on $f^{-1}\left( [0,M) \right)$ and strictly decreasing on $f^{-1}(M)$, where $M = \sup\operatorname{Im}(f)$.
\end{proof}

\section{Equilibrium decomposition for token clusters sharing one token} \label{app:decomposition}

In Subsection \ref{subsec:allowed-tokens} we explained why we expect to have a core-periphery trading topology. Peripheral tokens could form clusters, each connected to the core through a single core token. The following proposition extends \cite[Theorem 5]{Speedex} to this setting in the {\it supply function} framework of this paper and also establishes the converse.

\begin{Proposition} \label{prop:decomposition}
 Let $\PP_n \xto \varphi \R^n$ be a supply function, and let
 $I, J \subset I_n$ satisfy $I \cap J = \{k\}$.
 Assume $\varphi = \varphi_I + \varphi_J$ with $\varphi_I$ and $\varphi_J$
 supply functions supported on $I$ and $J$ respectively.

 (a) If $\varphi(p) = 0$, we have $\varphi_I(p) = 0$ and $\varphi_J(p) = 0$.

 (b) If $\varphi_I(a) = 0$, $\varphi_J(b) = 0$, with $a, b \in \PP_n$, we can form an equilibrium $p$ for $\varphi$ by merging $a$ and $b$ as follows:
 $$p_i = a_i / a_k \mbox{ for } i \in I$$
 $$p_j = b_j / b_k \mbox{ for } j \in J$$
 $$p_l = 1 \mbox{ for } l \notin I \cup J$$
\end{Proposition}

\begin{proof} $ $

(a) Let $p$ be such that $\varphi(p) = 0 = \varphi_I(p) + \varphi_J(p)$.
By Proposition \ref{prop:support}, for any $i \in I \setminus J$, it holds that $(\varphi_J)_i(p) = 0$.
 Thus $(\varphi_I)_i(p) = 0$ for $i \in I \setminus \{k\}$. Since $\varphi_I$ is a supply function, $\sum_{i \in I} p_i (\varphi_I)_i(p) = 0$. Hence $p_k (\varphi_I)_k(p) = 0$, and therefore $(\varphi_I)_k(p) = 0$, which completes the proof that $\varphi_I(p)=0$. Analogously, $\varphi_J(p)=0$.

 (b) $p$ is well defined because the first and second expressions give $p_k = 1$.
 Since $\varphi_I$ does not depend on the coordinates indexed by $J \setminus I$ and is homogeneous of degree $0$, we have $\varphi_I(p)=\varphi_I(a)=0$. Similarly, $\varphi_J(p) = 0$.
\end{proof}

{\bf Acknowledgments:}
I thank Flashbots for supporting this project during its initial stage, and Christoph Schlegel and Geoffrey Ramseyer for valuable feedback.

 \medskip
{\bf AI assistance: } I used Codex with GPT-6-Astra for language editing, mathematical exploration, and review of arguments and proofs during the preparation of this paper. I take full responsibility for the final content.

\bibliographystyle{amsplain}
\bibliography{mev}

\end{document}